\documentclass[a4paper,USenglish,cleveref, autoref, thm-restate]{lipics-v2019}

\usepackage{color}
\usepackage{xcolor}
\definecolor{dkgreen}{rgb}{0,0.6,0}
\definecolor{gray}{rgb}{0.5,0.5,0.5}
\definecolor{lgray}{rgb}{0.85,0.85,0.85}
\definecolor{mauve}{rgb}{0.58,0,0.82}
\definecolor{red}{rgb}{1.0,0,0}
\def\colorno{black}
\def\colormaybe{white}
\def\coloryes{gray}
\def\colorconnect{blue}

\usepackage{amsmath,amsfonts,amssymb}

\DeclareMathOperator*{\bigkringel}{\scalerel*{\circ}{\sum}}

\usepackage{pgfplots}
\usepackage{scalerel}
\usepackage{subcaption}
\usepackage{lcg}
\usepackage{xfp}
\usepackage{hyperref}
\usepackage{cleveref}
\usepackage[ruled,vlined,linesnumbered]{algorithm2e}
\usepackage{thmtools,thm-restate}

\usetikzlibrary{fit}

\title{The Complexity of Online Graph Games}

\titlerunning{The Complexity of Online Graph Games}

\author{Janosch Fuchs}{Department of Computer Science, RWTH Aachen University, Germany}{fuchs@algo.rwth-aachen.de}{https://orcid.org/0000-0003-3993-222X}{}

\author{Christoph Grüne}{Department of Computer Science, RWTH Aachen University, Germany}{gruene@algo.rwth-aachen.de}{https://orcid.org/0000-0002-7789-8870}{}

\author{Tom Janßen}{Department of Computer Science, RWTH Aachen University, Germany}{janssen@algo.rwth-aachen.de}{https://orcid.org/0000-0003-4617-3540}{}

\authorrunning{J. Fuchs, C. Grüne, T. Janßen}

\Copyright{Janosch Fuchs, Christoph Grüne and Tom Janßen}



\category{}

\relatedversion{}

\supplement{}

\funding{This work is funded by the Deutsche Forschungsgemeinschaft (DFG, German Research Foundation) –- GRK 2236/1, WO 1451/2-1.}

\acknowledgements{}

\nolinenumbers 

\hideLIPIcs  

\EventEditors{John Q. Open and Joan R. Access}
\EventNoEds{2}
\EventLongTitle{42nd Conference on Very Important Topics (CVIT 2016)}
\EventShortTitle{CVIT 2016}
\EventAcronym{CVIT}
\EventYear{2016}
\EventDate{December 24--27, 2016}
\EventLocation{Little Whinging, United Kingdom}
\EventLogo{}
\SeriesVolume{42}
\ArticleNo{23}

\begin{document}

\newenvironment{longversion}{}{}
\newenvironment{shortversion}{}{}

\newcommand{\papershort}{
	\excludecomment{longversion}
	\includecomment{shortversion}
}
\newcommand{\paperlong}{
	\excludecomment{shortversion}
}

\newcommand{\NP}{\textsl{NP}}
\newcommand{\PSPACE}{\textsl{PSPACE}}
\newcommand{\PTIME}{\textsl{PTIME}}

\newcommand{\indexSet}[2]{\ensuremath{~#1 \in \{1, \dots, #2\}}}

\newcommand{\CombinatorialProblem}{\ensuremath{P^{C}}}
\newcommand{\CombinatorialGraphProblem}{\ensuremath{P^{CG}}}
\newcommand{\OnlineCombinatorialGraphProblem}{\ensuremath{P^{CG}_o}}
\newcommand{\VertexCoveringGraphProblem}{\ensuremath{P^{VS}}}
\newcommand{\OnlineVertexCoveringGraphProblem}{\ensuremath{P^{VS}_o}}
\newcommand{\VertexPartitioningGraphProblem}{\ensuremath{P^{VPG}}}
\newcommand{\OnlineVertexPartitioningGraphProblem}{\ensuremath{P^{VPG}_o}}
\newcommand{\VertexOrderingGraphProblem}{\ensuremath{P^{VOG}}}
\newcommand{\OnlineVertexOrderingGraphProblem}{\ensuremath{P^{VOG}_o}}

\newcommand{\LiteralGadget}[1]{\ensuremath{G_{\ell_{#1}}}}
\newcommand{\NegLiteralGadget}[1]{\ensuremath{G_{\overline{\ell}_{#1}}}}
\newcommand{\VariableGadget}{\ensuremath{G_{var}}}
\newcommand{\ClauseGadget}{\ensuremath{G_{c}}}
\newcommand{\ExtensionGadget}{\ensuremath{G_{ext}}}
\newcommand{\IdGadget}{\ensuremath{G_{id}}}
\newcommand{\FakeClauseGadget}{\ensuremath{G_{\hspace{-0.5mm}f\hspace{-0.4mm}c}}}
\newcommand{\DependencyRevealGadget}{\ensuremath{G_{dr}}}

\newcommand{\openNeighborhood}[2]{\ensuremath{N_{#1}(#2)}}
\newcommand{\closedNeighborhood}[2]{\ensuremath{N_{#1}[#2]}}
\newcommand{\NeighborhoodSubgraph}[2]{\ensuremath{G^N_{#1}[#2]}}

\newcommand{\todo}[1]{\textcolor{red}{\textbf{TODO: #1}}\\}
\newcommand{\todonn}[1]{\textcolor{red}{\textbf{TODO: #1}}}
\newcommand{\todolater}[1]{\textcolor{orange}{\textbf{TODO (later): #1}}\\}
\newcommand{\todolayout}[1]{\textcolor{violet}{\textbf{TODO (layout): #1}}\\}
\newcommand{\cit}{\textsuperscript{[citation needed]}}

\newcommand{\sat}{\textsc{3-Satisfiability}}
\newcommand{\qsat}{\textsc{TQBF}}
\newcommand{\qsatgame}{\textsc{TQBF Game}}

\newcommand{\ds}{\textsc{Dominating Set}}
\newcommand{\ods}{\textsc{Online Dominating Set}}
\newcommand{\odn}{\textsc{Online Dominating Set Number}}
\newcommand{\odg}{\textsc{Online Dominating Set Game}}

\newcommand{\vc}{\textsc{Vertex Cover}}
\newcommand{\ovc}{\textsc{Online Vertex Cover}}
\newcommand{\ovn}{\textsc{Online Vertex Cover Number}}
\newcommand{\ovg}{\textsc{Online Vertex Cover Game}}

\newcommand{\is}{\textsc{Independent Set}}
\newcommand{\ois}{\textsc{Online Independent Set}}
\newcommand{\oin}{\textsc{Online Independent Set Number}}
\newcommand{\oig}{\textsc{Online Independent Set Game}}

\newcommand{\cn}{\textsc{Chromatic Number}}

\newcommand{\hp}{\textsc{Hamiltonian Path}}

\paperlong

\maketitle

\begin{abstract}
Online computation is a concept to model uncertainty where not all information on a problem instance is known in advance. 
An online algorithm receives requests which reveal the instance piecewise and has to respond with irrevocable decisions. 
Often, an adversary is assumed that constructs the instance knowing the deterministic behavior of the algorithm.
Thus, the adversary is able to tailor the input to any online algorithm.
From a game theoretical point of view, the adversary and the online algorithm are players in an asymmetric two-player game. 

To overcome this asymmetry, the online algorithm is equipped with an isomorphic copy of the graph, which is referred to as unlabeled map.
By applying the game theoretical perspective on online graph problems, where the solution is a subset of the vertices, we analyze the complexity of these online vertex subset games.
For this, we introduce a framework for reducing online vertex subset games from \qsat.
This framework is based on gadget reductions from \sat{} to the corresponding offline problem.
We further identify a set of rules for extending the \sat{}-reduction and provide schemes for additional gadgets which assure that these rules are fulfilled.
By extending the gadget reduction of the vertex subset problem with these additional gadgets, we obtain a reduction for the corresponding online vertex subset game.

At last, we provide example reductions for online vertex subset games based on \vc{}, \is{}, and \ds{}, proving that they are \PSPACE-complete. 
Thus, this paper establishes that the online version with a map of \NP-complete vertex subset problems form a large class of \PSPACE-complete problems. 
\end{abstract}

\newpage

\definecolor{darkgreen}{rgb}{0.0, 0.5, 0.21}
\section{Introduction}\label{intro:sec:intro}

Online computation is an intuitive concept to model real time computation where the full instance is not known beforehand.
In this setting, the instance is revealed piecewise to the online algorithm and each time a piece of information is revealed, an irrevocable decision by the online algorithm is required. 
To analyze the worst-case performance of an algorithm solving an online problem, a malicious adversary is assumed.

The adversary constructs the instance while the online algorithm has to react and compute a solution. 
This setting is highly asymmetric in favor of the adversary. 
Thus, for most decision problems, the adversary is able to abuse the imbalance of power to prevent the online algorithm from finding a solution that is close to the optimal one.  
To overcome the imbalance, there are different extensions of the online setting in which the online algorithm is equipped with some form of a priori knowledge about the instance. 
In this work, we analyze the influence of knowing an isomorphic copy of the input instance, which is also called an \emph{unlabeled map}. 
With the unlabeled map, the algorithm is able to recognize unique structures while the online instance is revealed -- like a vertex with unique degree -- but it cannot distinguish isomorphic vertices or subgraphs.  

The relation between the online algorithm and the adversary corresponds to players in an asymmetric two-player game, in which the algorithm wants to maximize its performance and the adversary's goal is to minimize it. 
The unlabeled map can be considered as the game board. 
One turn of the game consists of a move by the adversary followed by a move of the online algorithm. 
Thereby, the adversary reveals a vertex together with its neighbors and the online algorithm has to irrevocably decide whether to include this vertex in the solution or not. 
The problem is to evaluate whether the online algorithm has a \emph{winning strategy}, that is, it is able to compute a solution of size smaller/greater or equal to the desired solution size $k$, 
for all possible adversary strategies. 


Papadimitriou and Yannakakis make use of the connection between games and online algorithms for analyzing the canadian traveler problem in~\cite{DBLP:journals/tcs/PapadimitriouY91}, which is an online problem where the task is to compute a shortest $s$-$t$-path in an a priori known graph in which certain edges can be removed by the adversary. 
They showed that the computational problem of devising a strategy that achieves a certain competitive ratio is \PSPACE-complete by giving a reduction from \textsc{True Quantified Boolean Formula}, short \qsat.

Independently, Halld{\'{o}}rsson~\cite{DBLP:journals/combinatorics/Halldorsson00} introduced the problems online coloring and online independent set on a priori known graphs, which is equivalent to having an unlabeled map. 
He studies how the competitive ratios improve compared to the model when the graph is a priori not known. 
Based on these results, Halld{\'{o}}rsson et al.~\cite{DBLP:journals/tcs/HalldorssonIMT02} continued the work on the online independent set problem without a priori knowing the graph. 
These results are then applied by Boyar et al.~\cite{DBLP:journals/mst/BoyarFKM17} to derive a lower bound for the advice complexity of the online independent set problem. Furthermore, they introduce the class of asymmetric online covering problems (AOC) containing \ovc{}, \ois{}, \ods{} and others. 
Boyar et al.~\cite{DBLP:journals/dam/BoyarK18} 
analyze the complexity of these problems as graph property, namely the online vertex cover number, online independence number and online domination number, by showing their \NP-hardness.

Moreover based on the work by Halld{\'{o}}rsson~\cite{DBLP:journals/combinatorics/Halldorsson00}, Kudahl~\cite{DBLP:conf/ciac/Kudahl15} shows \PSPACE-completeness of the decision problem \textsc{Online Chromatic Number with Precoloring} on an a priori known graph, which asks whether some online algorithm is able to color $G$ with at most $k$ colors for every possible order in which $G$ is presented while having a precolored part in $G$.
This approach is then improved by B{\"{o}}hm and Vesel{\'{y}}~\cite{DBLP:journals/mst/BohmV18} by showing that \textsc{Online Chromatic Number} is \PSPACE-complete by giving a reduction from \qsat. 
 




Our contribution is to analyze the computational complexity of a subclass of AOC problems that consider graph problems where the solution is a subset of the vertices. 
Similar to the problem \textsc{Online Chromatic Number}, we equip the online algorithm with an unlabeled map in order to apply and formalize the ideas of B{\"{o}}hm and Vesel{\'{y}}.
We call these problems \emph{online vertex subset games} due to their relation to two-player games. 
While symmetrical combinatorial two-player games are typically \PSPACE-complete~\cite{FRAENKEL198721}, this principle does not apply to our asymmetrical setting. 
We are still able to prove \PSPACE-completeness for the online vertex subset games based on \vc, \is{} and \ds{} by designing reductions such that the adversary's optimal strategy corresponds to the optimal strategy of the $\forall$-player in TQBF.

In order to derive reductions from TQBF to online vertex subset games, we identify properties describing the revelation or concealment of information to correctly simulate the $\forall$- and $\exists$-decisions as well as the evaluation of the quantified Boolean formula in the online vertex subset game.
This simulation is modeled by disjoint and modular gadgets, which form a so-called gadget reduction -- similar to already known reductions between NP-complete problems.
Different forms of gadget reductions are described by Agrawal et al.~\cite{DBLP:conf/stoc/AgrawalAIPR97} who formalize $AC^0$-gadget-reductions in the context of \NP-completeness and by Trevisan et al.~\cite{DBLP:conf/focs/TrevisanSSW96} who describe gadgets in reductions of problems that are formalized as linear programs.
By formalizing gadgets capturing the above mentioned properties, we provide a framework to derive reductions for other online vertex subset games, which are based on problems that are gadget-reducible from \sat{}.

\paragraph*{Paper Outline}\label{intro:subsec:paperOutline}
First, we explain the online setting that we use throughout the paper and important terms, e.g., the online game, the problem class and the reveal model of our online problems.
Secondly, we define the gadget reduction framework to reduce \sat{} to vertex subset problems. 
In the third section, we extend the framework to the online setting by identifying a set of important properties that must be fulfilled in the reduction. 
We also provide a scheme for gadgets that enforce these properties to generalize the framework to arbitrary vertex subset problems. 
In the fourth section, we detail the application of this framework to the problem \vc.
Lastly, we apply the framework to the problems \is{} and \ds{} in the fifth section.
At the end, we summarize the results and give a prospect on future possible work.



\paragraph*{Neighborhood Reveal Model}
Each request of the online problem reveals information about the instance for the online algorithm. 
The amount of information in each step is based on the reveal model. 
For an online problem with a map, the subgraph that arrives in one request is called \textit{revelation subgraph}. 

The neighborhood reveal model, which we use in this paper, was introduced by Harutyunyan et. al.~\cite{DBLP:conf/mfcs/HarutyunyanPR21}.
Within that model, the online algorithm gains information about the complete neighborhood of the revealed vertex. 
Nevertheless, the online algorithm has to make a decision on the current revealed vertex only but not on the exposed neighborhood vertices.
All exposed but not yet revealed neighborhood vertices have to be revealed in the process of the online problem such that a decision can be made upon them.
We denote the closed neighborhood of $v$ with $N[v]$, that is, the set of $v$ and all vertices adjacent to $v$. 

\begin{definition}[Neighborhood Reveal Model]\label{intro:def:neighborhoodRevealModel}
The neighborhood reveal mo\-del is defined by an ordering of graphs $(V_i, A_i, E_i)_{i \leq |V|}$. The reveal order of the adversary is defined by $adv \in S_{|V|}$, where $S_{|V|}$ is the symmetric group of size $|V|$. 
The graph $G_i$ is defined by 
\begin{align*}
	V_0 &= E_0 = \emptyset, &  \\
	V_i &= V_{i-1} \cup N[{v_{adv(i)}}], \ &\text{for} \ 0 < i \leq |V| & \\
	E_i &= E_{i-1} \cup \{(v_{adv(i)}, w) \in E \ \mid w \in V_i \}, \ &\text{for} \ 0 < i \leq |V|&.
\end{align*}
The \textit{revelation subgraph} $G^\prime$ in the neighborhood reveal model is the subgraph of $G_i$ defined by $G' = (V', E')$ with $V' = N[v_{adv(i)}]$ and $E' = \{ \{v_{adv(i)}, w\} \in E\}$. 
The online algorithm has to decide whether $v_{adv(i)}$ is in the solution or not.
\end{definition}

\paragraph*{Online Vertex Subset Games}
Throughout the paper, we consider a special class of combinatorial graph problems.
The question is to find a vertex subset, whereby the size should be either smaller or equals, for minimization problems, or greater or equals, for maximization problems, some $k$, which is part of the input.
Thereby, the vertex set needs to fulfill some constraints based on the specific problem.
We call these problems vertex subset problems.
Well-known problems like \vc, \is{} and \ds{} are among them.

We denote the online version with a map of a vertex subset problem \VertexCoveringGraphProblem{} with \OnlineVertexCoveringGraphProblem{} and define them as follows.

\begin{definition}[Online Vertex Subset Game]\label{intro:ex:vertexCoveringGraphProblem}
	An \textit{online vertex subset game} $\OnlineVertexCoveringGraphProblem{}$ has a graph $G$ and a $k \in \mathbb{N}$ as input.
	The question is, whether the online algorithm is able to find a vertex set of size smaller (resp. greater) or equals $k$, which fulfills the constraints of $\VertexCoveringGraphProblem$ for all strategies of the adversary.
	Thereby, the online algorithm has access to an isomorphic copy of $G$ and the adversary reveals the vertices according to the neighborhood reveal model.
\end{definition}

\section{Gadget Reductions}

\textit{Gadget reductions} are a concept to reduce combinatorial problems in a modular and structured way.
For the context of the paper, we define gadget reductions from \sat{} to vertex subset problems. The
\sat{} instances $\varphi = (L, C)$ are defined by their literals $L$ and their clauses $C$.
We use a literal vertex $v_\ell$ for all $\ell \in L$ to represent a literal in the graph.
There are implicit relations over the literals besides the explicit relation $C$, in that the reduction may be decomposed.
For example, the relation between a literal and its negation, which is usually implicitly used to build up a \textit{variable gadget}. 
These variable gadgets are connected by graph substructures that assemble the clauses as \textit{clause gadgets}.

\begin{definition}[Gadget Reduction from \sat{} to Vertex Subset Problems]\label{intro:def:gadgetReduction}
	A gadget reduction $R_{gadget}(\VertexCoveringGraphProblem)$	
	from a \sat{} formula $\varphi = (L, C)$ to a vertex subset problem with graph $G_\varphi = (V, E)$ is a tuple containing functions from the literal set and all relations of the \sat{} formula to the vertex set and all relations of the vertex subset problem. 
	In the following, we denote the gadget based on element $x$ to be $G_x := (V_x, E_x)$ with $V_x$ being a set of vertices and $E_x$ a set of edges, whereby the edges are potentially incident to vertices of a different gadget.
	
	The literal set of a \sat{} formula $\ell_1, \ell_2, \ldots, \ell_{|L|-1}, \ell_{|L|}$ is mapped to the vertex set of the graph problem.
	Thereby, each literal is mapped to exactly one vertex: $$R^{L \rightarrow V}_{gadget}(\VertexCoveringGraphProblem): L \rightarrow V, \ell \mapsto G_\ell $$
	
	The following relations on the literals are mapped as well.\\
	\begin{tabular}{lll}
		(1) & Literal - Negated Literal: &$
		R^{L,\overline{L}}_{gadget}(\VertexCoveringGraphProblem):
		R(L,\overline{L}) \rightarrow (V, E), (\ell, \overline{\ell}) \mapsto G_{\ell, \overline{\ell}}$ \\
		(2) & Clause: &$
		R^{C}_{gadget}(\VertexCoveringGraphProblem):
		R(C) \rightarrow (V, E), C_j \mapsto G_{C_j}$\\
		(3) & Literal - Clause: &$
		R^{L, C}_{gadget}(\VertexCoveringGraphProblem):
		R(L, C) \rightarrow (V, E), (\ell, C_j) \mapsto G_{\ell, C_j}$\\
		(4) & Negated Literal - Clause: &$
		R^{\overline{L}, C}_{gadget}(\VertexCoveringGraphProblem):
		R(\overline{L}, C) \rightarrow (V, E), (\overline{\ell}, C_j) \mapsto G_{\overline{\ell}, C_j}$
	\end{tabular}
	\newline
	Additionally, the following mapping allows for constant parts that do not change depending on the instance: $R^{const}_{gadget}(\VertexCoveringGraphProblem): \emptyset \rightarrow (V,E), \emptyset \mapsto G_{const}$.
	Thereby, the vertices and edges of all gadgets are pairwise disjoint.
\end{definition}

We use the more coarse grained view of \textit{variable gadgets} as well.
These combine the mappings $R^{L \rightarrow V}_{gadget}$ and $R^{L,\overline{L}}_{gadget}$ to $R^{X}_{gadget}$, and $R^{L, C}_{gadget}$ and $R^{\overline{L}, C}_{gadget}$ to $R^{X, C}_{gadget}$, where $X$ is the set of $n$ variables.

The important function of the variable gadget is to ensure that only one of the literals of $\ell, \overline{\ell} \in L$ is chosen.
On the other hand, the function of the clause gadget is to ensure together with the constraints of the vertex subset problem \VertexCoveringGraphProblem{} that the solution encoded on the literals fulfill the \sat-formula if and only if the literals induce a correct solution.
These functionalities are utilized in the correctness proof of the reduction by identifying the logical dependencies between the literal vertices $v_\ell$ for $\ell \in L$ and all other vertices based on the graph and the constraints of \VertexCoveringGraphProblem{} together with combinatorial arguments on the solution size.
We denote these logical dependencies as \textit{solution dependencies} as they are logical dependencies on the solutions of \VertexCoveringGraphProblem.
%
Due to the asymmetric nature of the online problems, the adversary can reveal a solution dependent vertex before revealing the corresponding literal vertex.
Thus, a decision on the solution dependent vertex is implicitly also a decision on the literal vertex, although it has not been revealed.
We address this specific problem later in the description of the framework.


\begin{definition}[Solution dependent vertices]\label{intro:def:solutionDependentVertices}
	Given a gadget reduction, the following vertices of the reduction graph are solution dependent:
	\begin{enumerate}
		\item All literal vertices are solution dependent on their respective variable.
		\item For a literal $\ell$ (resp. its negation $\overline{\ell}$), we denote the set of vertices that need to be part of the solution if $v_\ell$ (resp. $v_{\overline{\ell}}$) is part of the solution with $V_\ell$ (resp. $V_{\overline{\ell}}$).
		Then the vertices, that are in one but not both of these sets, i.e. $V_\ell \ \triangle \ V_{\overline{\ell}}$, are solution dependent on the corresponding variable. 
	\end{enumerate}
	All vertices that are not solution dependent on any variable are called \emph{solution independent}.
\end{definition}

For example, in the reduction from \sat{} to \vc~\cite{garey1979computers}, the following solution dependencies apply:
For each literal, the vertices $v_\ell$ and $v_{\overline{\ell}}$ are solution dependent on their respective variable.
Furthermore, if a literal is part of the solution, all clause vertices representing its negation must also be part of the solution.
Thus all clause vertices are solution dependent on their respective variable.
Consequently, all vertices of the reduction graph for vertex cover are solution dependent.

\section{A Reduction Framework for Online Vertex Subset Games}\label{generalReductionCovering:sec:aGeneralReductionFrameworkForVertexCoveringGraphProblems}

In this section, we present a general framework for reducing \qsatgame{} to an arbitrary online vertex subset game \OnlineVertexCoveringGraphProblem{} with neighborhood reveal model. 
The \qsatgame{} is played on a fully quantified Boolean formula, where one player decides the $\exists$-variables and the other decides the $\forall$-variables, in the order they are quantified.
Deciding whether the $\exists$-player has a winning strategy is \PSPACE-complete \cite{DBLP:conf/stoc/StockmeyerM73}, and thus this reduction proves \PSPACE-hardness for \OnlineVertexCoveringGraphProblem{}.
We assume that the \qsatgame{} consists of clauses with at most three literals, which is also known to be \PSPACE-complete \cite{garey1979computers}.

Before we describe the reduction, we prove that the online game version of each vertex covering graph problem in \NP{} is in \PSPACE.

\begin{restatable}{theorem}{generalReductionCoveringthmPspaceContainment}\label{generalReductionCovering:thm:PspaceContainment}
	If \VertexCoveringGraphProblem{} is in \NP{}, then \OnlineVertexCoveringGraphProblem{} is in \PSPACE.
\end{restatable}
\begin{longversion}
	\begin{proof}
		The instance graph is encoded in linear space.
		The solution (subset of vertices) is encoded in at most linear space because the base problem \VertexCoveringGraphProblem{} is in \NP{}.
		The number of moves is the number of universe elements, which are vertices.
		This is linear in the input size.
		For each move, the currently revealed graph is stored as well as the current solution.
		This is again linear in the input size.
		Thus, the used space is overall polynomial for each move, whereby there is only a linear number of moves.
		Consequently, the problem \OnlineVertexCoveringGraphProblem{} is in \PSPACE{}.
	\end{proof}
\end{longversion}

This framework uses an (existing) gadget reduction of the vertex subset problem \VertexCoveringGraphProblem{} from \sat{} and extends it in order to give the online algorithm the ability to recognize the current revealed vertex.
Due to the quantification of variables, we call the variable gadget of a $\forall$-variable a $\forall$-gadget (resp. $\exists$-gadget for an $\exists$-variable). 
Based on this, the online algorithm can use a one-to-one correspondence between the solution of the \qsatgame{} instance and the \OnlineVertexCoveringGraphProblem{} instance.
The one-to-one correspondence between the $\forall$-variables and the $\forall$-gadgets is ensured by the knowledge of the adversary about the deterministic online algorithm. 
It simulates the response of the algorithm on the $\forall$-gadget. 

\paragraph*{Extension Gadgets}

We extend the reduction graph $G_\varphi$ of the offline problem with gadgets to a reduction graph for the online problem.
These gadgets extend $G_\varphi$ by connecting to a subset of its vertices.
We denote these gadgets \ExtensionGadget{} as \textit{extension gadgets}.

\begin{definition}[Graph Extension]\label{generalReductionCovering:def:graphExtension}
	A graph extension of a graph $G = (V, E)$ by an extension gadget $\ExtensionGadget{} = (V_{ext}, E_{ext}, E_{con})$ with the set of connecting edges $E_{con} \subseteq V \times V_{ext}$ is defined as $H = G \circ \ExtensionGadget$, whereby
	\begin{align*}
		V(H) &= V \cup V_{ext}, \\
		E(H) &= E \cup E_{ext} \cup E_{con}. 
	\end{align*}
	We further define $G \bigkringel_{i \in I} G^i_{ext} := \left(\ldots\left(\left(G \circ G^{i_1}_{ext}\right) \circ G^{i_2}_{ext}\right) \circ \ldots \right)\circ G^{i_{|I|}}_{ext}$.
\end{definition}

We also need the notion of \emph{self-contained} gadgets.
These do not influence the one-to-one correspondence between solutions of the online vertex subset game \OnlineVertexCoveringGraphProblem{} and \qsatgame. 
In other words, optimal solutions on the graph and the extension gadget can be disjointly merged to obtain an optimal solution on the extended graph.
Due to this independence, we are able to provide local information to the online algorithm via the map without changing the underlying formula.
An example for self-contained extension gadgets is provided in \Cref{framework:fig:extensionGadget}.
Note that, it can occur that self-containment depends on the extended graph.

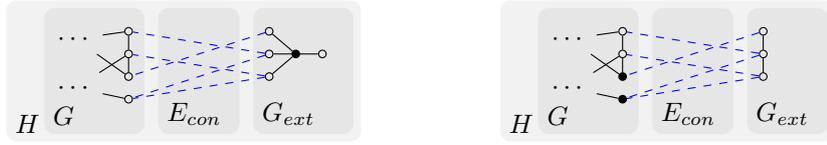
\begin{figure}[!ht]
	\centering
	\scalebox{1}{
	\begin{tikzpicture}[scale=0.5,
						node/.style = {shape=circle, draw, inner sep=0pt, minimum size=0.25cm},
						textnode/.style = {shape=circle, draw, inner sep=0pt, minimum size=0.4cm},
						smallnode/.style = {shape=circle, draw, inner sep=0pt, minimum size=0.1cm},
						box/.style = {rectangle, fill=gray!20, rounded corners, fill opacity=1, inner sep=1pt}]
		\node (H1box1) at (-1, 0.1) {};
		\node (H1box2) at (7.7, 3.2) {};
		\node[box, fill=gray!10, fit=(H1box1)(H1box2)] (H1box) {};
		\node[above right] (H1text) at (H1box.south west) {$H$};
		
		\node (G1box1) at (0, 0.3) {};
		\node (G1box2) at (2, 3) {};
		\node[box, fit=(G1box1)(G1box2)] (G1box) {};
		\node[above right] (G1text) at (G1box.south west) {$G$};
		
		\node[smallnode] (v11) at (1.9, 2.7) {};
		\node[smallnode] (v12) at (1.9, 2.1) {};
		\node[smallnode] (v13) at (1.9, 1.5) {};
		\node[smallnode] (v14) at (1.9, 0.9) {};
		\path[-] (v11) edge (v12) (v12) edge (v13);
		
		\node (G1dots1) at (0.5, 2.5) {$\cdots$};
		\node (G1dots2) at (0.5, 1.2) {$\cdots$};
		\path[-] (G1dots1) edge (v11) edge (v13) (G1dots2) edge (v14) edge (v12);
		
		\node (G1ebox1) at (5.5, 0.3) {};
		\node (G1ebox2) at (7.5, 3) {};
		\node[box, fit=(G1ebox1)(G1ebox2)] (G1ebox) {};
		\node[above right] (G1etext) at (G1ebox.south west) {$G_{ext}$};
		
		\node[smallnode] (ve11) at (5.6, 2.7) {};
		\node[smallnode] (ve12) at (5.6, 2.1) {};
		\node[smallnode] (ve13) at (5.6, 1.5) {};
		\node[smallnode, fill=black] (ve14) at (6.3, 2.1) {};
		\node[smallnode] (ve15) at (7, 2.1) {};
		\path[-] (ve14) edge (ve11) edge (ve12) edge (ve13) edge (ve15);
		
		\node (E1box1) at (3, 0.3) {};
		\node (E1box2) at (4.5, 3) {};
		\node[box, fit=(E1box1)(E1box2)] (E1box) {};
		\node[above right] (E1text) at (E1box.south west) {$E_{con}$};
		
		\path[-, draw=\colorconnect, dashed] (ve13) edge (v14) edge (v12);
		\path[-, draw=\colorconnect, dashed] (ve12) edge (v14) edge (v11);
		\path[-, draw=\colorconnect, dashed] (ve11) edge (v13);
		
		\node (H2box1) at (12, 0.1) {};
		\node (H2box2) at (20.2, 3.2) {};
		\node[box, fill=gray!10, fit=(H2box1)(H2box2)] (H2box) {};
		\node[above right] (H2text) at (H2box.south west) {$H$};
		
		\node (G2box1) at (13, 0.3) {};
		\node (G2box2) at (15, 3) {};
		\node[box, fit=(G2box1)(G2box2)] (G2box) {};
		\node[above right] (G2text) at (G2box.south west) {$G$};
		
		\node[smallnode] (v21) at (14.9, 2.7) {};
		\node[smallnode] (v22) at (14.9, 2.1) {};
		\node[smallnode, fill=black] (v23) at (14.9, 1.5) {};
		\node[smallnode, fill=black] (v24) at (14.9, 0.9) {};
		\path[-] (v21) edge (v22) (v22) edge (v23);
		
		\node (G2dots1) at (13.5, 2.5) {$\cdots$};
		\node (G2dots2) at (13.5, 1.2) {$\cdots$};
		\path[-] (G2dots1) edge (v21) edge (v23) (G2dots2) edge (v24) edge (v22);
		
		\node (G2ebox1) at (18.5, 0.3) {};
		\node (G2ebox2) at (20, 3) {};
		\node[box, fit=(G2ebox1)(G2ebox2)] (G2ebox) {};
		\node[above right] (G2etext) at (G2ebox.south west) {$G_{ext}$};
		
		\node[smallnode] (ve21) at (18.6, 2.7) {};
		\node[smallnode] (ve22) at (18.6, 2.1) {};
		\node[smallnode] (ve23) at (18.6, 1.5) {};
		\path[-] (ve22) edge (ve21);
		\path[-] (ve22) edge (ve23);
		
		\node (E2box1) at (16, 0.3) {};
		\node (E2box2) at (17.5, 3) {};
		\node[box, fit=(E2box1)(E2box2)] (E2box) {};
		\node[above right] (E2text) at (E2box.south west) {$E_{con}$};
		
		\path[-, draw=\colorconnect, dashed] (ve23) edge (v24) edge (v22);
		\path[-, draw=\colorconnect, dashed] (ve22) edge (v24) edge (v21);
		\path[-, draw=\colorconnect, dashed] (ve21) edge (v23);
	\end{tikzpicture}
}
	\caption{On the left, there is an example for an extension gadget that is self-contained w.r.t. the dominating set problem: No matter the solution on $G$, at least one vertex of $G_{ext}$ has to be chosen. Additionally, choosing the black vertex of $G_{ext}$ dominates all vertices attached to $G$, and thus any solution on $G$ remains valid.
	On the right, there is an example for an extension gadget that is not self-contained w.r.t. the dominating set problem: If the solution on $G$ contains the black vertices, it is also a solution for $H$, but the optimal solution on $G_{ext}$ contains one vertex.}
	\label{framework:fig:extensionGadget}
\end{figure}

For our reduction framework, we introduce three types of self-contained extension gadgets: fake clause gadgets, dependency reveal gadgets and ID gadgets.
The goal of these gadgets is that it is optimal for the adversary to reveal variables in the order of quantification, and that the online algorithm is able to assign the value of the $\exists$-variables, while the adversary is able to assign the value of the $\forall$-variables.

\paragraph*{Fake Clause Gadgets}
The number of occurrences of a certain literal in clauses is information that may allow the online algorithm to distinguish the literals of some $\forall$-variables, allowing the online algorithm to decide the assignment instead of the adversary.
To avoid this information leak, we add gadgets for all possible non-existing clauses to the reduction graph.
A \emph{fake clause gadget} is only detectable if and only if a vertex, which is part of that clause gadget, is revealed by the adversary.
The gadget needs to be self-contained such that the one-to-one correspondence between the solutions of the \OnlineVertexCoveringGraphProblem{} and \qsatgame{} is not affected.

\begin{definition}[Fake Clause Gadget]\label{generalReductionCovering:def:fakeClauseGadget}
	A fake clause gadget $\FakeClauseGadget(C'_j)$ for a non-existing clause $C'_j \notin C$ is an extension gadget that is self-contained.
	The fake clause gadgets are connected to the variable gadgets like the clause gadgets are to the variable gadgets according to the original gadget reduction, see \Cref{intro:def:gadgetReduction}.
\end{definition}

All fake clause gadgets are pairwise disjoint.
Let $G_\varphi$ be the gadget reduction graph and
$$
	G_\varphi' := G_\varphi \bigkringel_{C'_j \notin C} \FakeClauseGadget(C'_j).
$$
After adding fake clause gadgets for all clauses $C'_j \notin C$ to $G_\varphi$, the revelation subgraphs of all vertices $v \in V(G_\varphi')$, which are part of a literal gadget, are pairwise isomorphic.

\paragraph*{Dependency Reveal Gadgets}
The two functions of the dependency reveal gadgets are that the adversary chooses the reveal order to be the order of quantification and the online algorithm knows the decision on the $\forall$-variables after the decision is made by the adversary.
If the adversary deviates from the quantification order, the $\forall$-decision degenerates to an $\exists$-decision for the online algorithm.
On the other hand, since the adversary forces the online algorithm to blindly choose the truth value of a $\forall$-quantified variable, the online algorithm does not know the chosen truth value.
Thus, we need to reveal the truth value to the online algorithm whenever a solution dependent vertex of the next variables is revealed.

\begin{definition}[Dependency Reveal Gadget]\label{generalReductionCovering:def:dependencyRevealGadget}
	A dependency reveal gadget $\DependencyRevealGadget(x_i)$ for $\forall$-variable $x_i$ is an extension gadget that is self-contained with the property:
	Let $\ell, \overline{\ell}$ be the literals of $x_i$.
	If a solution dependent vertex of $x_j$ with $j \geq i$ is revealed to the online algorithm, the online algorithm is able to uniquely identify the vertices $v_{\ell}$ and $v_{\overline{\ell}}$.
\end{definition}

\paragraph*{ID Gadgets}
At last, the online algorithm needs information on the currently revealed vertex to identify it with the help of the map.
For this, we introduce ID gadgets, which make the revelation subgraph of vertices distinguishable to a certain extent.
Thus, the online algorithm is able to correctly encode the TQBF solution into the solution of the vertex subset game.
The ID gadget is always connected exactly to the vertex it identifies, thus they are pairwise disjoint.

\begin{definition}[ID Gadget]\label{generalReductionCovering:def:ID-Gadget}
	An identification gadget $\IdGadget(v)$ is a self-contained extension gadget connected to $v$ such that the revelation subgraph of $v$ is isomorphic to revelation subgraphs of vertices within a distinct vertex set $V^\prime \subseteq V$.
\end{definition}

\paragraph*{The General Reduction for Online Vertex Subset Games}

With the gadget schemes defined above, we are able to construct a gadget reduction from \qsatgame{} to \OnlineVertexCoveringGraphProblem{}.
The idea of the reduction is to construct the optimal game strategy for the online algorithm to compute the solution to the \qsatgame{} formula. Furthermore, encoding the solution to the \qsatgame{} into the \OnlineVertexCoveringGraphProblem{} instance is a winning strategy by using the equivalence of the $\exists$- and $\forall$-gadgets to the $\exists$- and $\forall$-variables. At last, there is a one-to-one correspondence between the reduction graph solution of \VertexCoveringGraphProblem{} and the \sat-solution.

A gadget reduction from \sat{} to vertex subset problem \VertexCoveringGraphProblem{} can be extended such that \OnlineVertexCoveringGraphProblem{} is reducible from \qsatgame{} as follows.
Recall that $G_\varphi$ is the gadget reduction graph of a fixed but arbitrary instance of \VertexCoveringGraphProblem.
\begin{enumerate}
	\begin{minipage}[t]{0.52\textwidth}
	\item Add fake clause gadgets for all clauses that are not in the \qsatgame{} instance
		$$
			G_\varphi^\prime = G_\varphi \bigkringel_{c' \notin C} \FakeClauseGadget(c')\text{ .}
		$$
	\end{minipage}
	\hfill
	\begin{minipage}[t]{0.36\textwidth}
	\item Add dependency reveal gadgets for all $\forall$-variables $x$
		$$
			G_\varphi^{\prime\prime} = G_\varphi^\prime \bigkringel_{\stackrel{x \in X}{x \text{ is } \forall}} \DependencyRevealGadget(x)\text{ .}
		$$
	\end{minipage}
	\item Add ID gadgets to all vertices 
	$$
	G_\varphi^{\prime\prime\prime} = G_\varphi^{\prime\prime} \bigkringel_{v \in V(G_\varphi'')} \IdGadget(v) \text{ .}
	$$
\end{enumerate}

Then, if all gadgets can be constructed in polynomial time, $G_\varphi^{\prime\prime\prime}$ is the corresponding reduction graph of \OnlineVertexCoveringGraphProblem.
The gadget reduction also implies the following gadget properties, which individually have to be proven for a specific problem.
\begin{enumerate}
	\item The fake clause gadgets are self-contained. 
	\item The dependency reveal gadgets are self-contained.
	\item The ID gadgets are self-contained.
	\item In $G_\varphi^{\prime\prime\prime}$, each solution dependent vertex which is not in a literal gadget of a $\forall$-variable has a unique revelation subgraph.
	\item In $G_\varphi^{\prime\prime\prime}$, the two literal vertices of a $\forall$-variable have the same revelation subgraph, but different from vertices of any other gadget.
	\item In $G_\varphi^{\prime\prime\prime}$, each vertex that is solution independent or part of an extension gadget has a revelation subgraph that allows for an optimal decision.
\end{enumerate}

From the above construction, the following \Cref{generalReductionCovering:lem:one-to-oneCorrespondence,generalReductionCovering:lem:optimalGamestrategyAdversary,generalReductionCovering:lem:decisionOnVariableGadget}, are fulfilled such that \OnlineVertexCoveringGraphProblem{} is proven to be \PSPACE-hard in the following \Cref{generalReductionCovering:thm:PSPACECompleteness}.

\begin{restatable}{lemma}{generalReductionCoveringlemonetooneCorrespondence}\label{generalReductionCovering:lem:one-to-oneCorrespondence}
	In the construction of the reduction, there is a one-to-one correspondence between the solution of the problem \OnlineVertexCoveringGraphProblem{} and \qsatgame, if there is a one-to-one correspondence between the solutions in the gadget reduction from \VertexCoveringGraphProblem{} and \sat.
	The equivalence is computable in \PTIME.
\end{restatable}
\begin{longversion}
	\begin{proof}
		The one-to-one-correspondence is preserved by the definition of self-con\-tained extension gadgets.
		All graph extensions are based on self-contained gadgets.
		Thus, the original solution is preserved and only complemented by the disjoint partial solution on all extension gadgets.
		
		The ID gadgets ensure that the optimal decision for vertices of the same degree is unique (by Gadget Property $4$-$6$), except for literal vertices of $\forall$-variables.
		Thus, a one-to-one correspondence between vertices of the map and the actual online game instance is easy to find by the online algorithm.
		Therefore, the online algorithm is able to decide whether to put a vertex in the solution or not by vertex degree for all vertices, except for literal vertices of $\forall$-variables.
		On the other hand, the adversary can decide the assignment of $\forall$-variables, as it is able to simulate the online algorithm and predict its decision.
	\end{proof}
\end{longversion}

In the following, we show that the adversary has to reveal one literal vertex of each variable gadget before revealing vertices of other gadgets (except ID gadgets). 
Furthermore, the adversary has to adhere to the quantification order of the variables when revealing the first literal vertices of each gadget.
If the adversary deviates from this strategy, it may allow the online algorithm to decide the truth assignment of $\forall$-variables.
This may allow the algorithm to win a game based on an unsatisfiable formula.
Thus, an optimal adversary strategy always follows the quantification order.
\begin{restatable}{lemma}{generalReductionCoveringlemoptimalGamestrategyAdversary}\label{generalReductionCovering:lem:optimalGamestrategyAdversary}
	Every optimal game strategy for the adversary adheres to the reveal ordering
	\begin{align}
		\LiteralGadget{1} \ or \ \NegLiteralGadget{1} < \LiteralGadget{2} \ or \ \NegLiteralGadget{2} < \cdots &< \LiteralGadget{n} \ or \ \NegLiteralGadget{n},&  \\
		\LiteralGadget{} \ or \ \NegLiteralGadget{} &< \ClauseGadget(C_j), &\text{ for all }\ell \in C_j \in C,  \\
		\LiteralGadget{} \ or \ \NegLiteralGadget{} &< \FakeClauseGadget(C'_j), &\text{ for all }\ell \in C'_j \notin C,  \\
		\LiteralGadget{} \ or \ \NegLiteralGadget{} &< \DependencyRevealGadget(x), &\text{ for all }x \in X. 
	\end{align}
\end{restatable}
\begin{longversion}
	\begin{proof}
		We prove each proposition one after another.
		\begin{enumerate}
			\item[(1)] Assume a vertex of the variable gadget of $x_j$ is revealed before any vertex of the variable gadget of $x_i$ is revealed for $i < j$.
			The following cases may apply:
			\begin{enumerate}
				\item \textit{$x_i$ and $x_j$ are $\exists$-variables}\\
				Then, a $\forall$-variable $x_k$ exists, $i < k < j$, for which the dependency reveal gadget is revealed before the vertices of its variable gadget are revealed.
				Therefore, the variable $x_k$ degenerates to an $\exists$-variable.
				\item \textit{$x_i$ is an $\exists$-variable and $x_j$ is an $\forall$-variable}
				\begin{itemize}
					\item[] Case 1: $j > i+1$
					Then, a $\forall$-variable $x_k$ exists, $i < i+1 \leq k < j$, for which the dependency reveal gadget is revealed before the vertices of its variable gadget are revealed.
					Therefore, the variable $x_k$ degenerates to an $\exists$-variable.
					\item[] Case 2: $j = i+1$
					Then, the decision on $x_j$ happens before $x_i$.
					However, the online algorithm is not able to detect which decision took place.
					Consequently, it has no additional information for the $\exists$-variable $x_i$.
					This does not change the game at all.
				\end{itemize}
				\item \textit{$x_i$ is an $\forall$-variable and $x_j$ is an $\exists$-variable}\\
				Then, the dependency reveal gadget for $x_j$ is revealed and $x_j$ degenerates to an $\exists$-variable.
				\item \textit{$x_i$ and $x_j$ are $\forall$-variables}\\
				Then, the dependency reveal gadget for $x_j$ is revealed and $x_j$ degenerates to an $\exists$-variable.
			\end{enumerate}
			\item[(2)] By revealing a clause gadget before the corresponding literal gadgets,
			it is revealed which literals are in the clause and whether the clause is a fake clause or not.
			Thus, it is dominant to reveal that information after revealing the literals.
			\item[(3)] By revealing a fake clause gadget before the corresponding literal gadgets,
			it is revealed which literals are in the fake clause and whether the fake clause is a fake clause or not.
			Thus, it is dominant to reveal that information after revealing the literals gadget.
			\item[(4)] The dependency reveal gadgets connected to a literal gadget reveal the information which literal gadget represents the negated literal. Thus, the $\forall$-variable degenerates to an $\exists$-variable.
		\end{enumerate}
		
		The degeneration of a $\forall$-variable to an $\exists$-variable gives the online algorithm the possibility to satisfy a possibly unsatisfiable formula, because revelation of the dependency reveal gadget is easily detectable in \PTIME.
	\end{proof}
\end{longversion}

\begin{restatable}{lemma}{generalReductionCoveringlemdecisionOnVariableGadget}\label{generalReductionCovering:lem:decisionOnVariableGadget}
	The vertex assignments of an $\exists$-variable gadget (resp. $\forall$-variable gadget) are equivalent to the decision of the $\exists$-player (resp. $\forall$-player) on an $\exists$-quantifier (resp. $\forall$-quantifier) in the \qsatgame{}.
	The equivalence is computable in \PTIME.
\end{restatable}
\begin{longversion}
	\begin{proof}
		By the dominating strategy of the adversary described in \Cref{generalReductionCovering:lem:optimalGamestrategyAdversary}, for all $i, 1 \leq i \leq n$, the adversary reveals all vertices of the variable gadget of variable $x_i$.
		Thereby, the online algorithm has to take a decision after each revealed vertex.
		\begin{itemize}
			\item[($\exists$)] Due to Gadget Properties $4$ and $6$, the online algorithm is able to detect which exact vertex of the variable gadget is revealed.
			Thus, the online algorithm is able to decide which vertices to take into the solution to encode both of the truth values of the variable into the solution.
			This implies that the online algorithm takes the decision on the variable as in the \qsatgame{}.
			\item[($\forall$)] Because of \Cref{generalReductionCovering:lem:optimalGamestrategyAdversary}, the adversary will reveal the literal gadgets first.
			The literal gadgets of different variables may only be connected by a path of length $\geq 2$, if there is a connection via a clause gadget or dependency reveal gadget.
			These do not reveal information as every possible clause is covered either by a clause gadget or a fake clause, but it is not revealed whether the connection is established by clause or fake clause gadget unless a vertex of a clause gadget or fake clause gadget is revealed.
			The dependency reveal gadget reveals only additional information if a reveal ordering which does not correspond to \Cref{generalReductionCovering:lem:optimalGamestrategyAdversary} is used.
			Due to Gadget Property $5$, the online algorithm is not able to detect whether a vertex that encodes an assignment to true or false is revealed over the degree.
			Thus, the online algorithm is not able to detect which literal gadget resembles the true or false assignment.
			
			Therefore, a reveal ordering of the variable gadgets of the $\forall$-variables exists that forces every fixed deterministic online algorithm to choose the options that prevent the online algorithm from winning if and only if the \qsatgame{} formula is unsatisfiable.
		\end{itemize}
		The computation is in \PTIME{} because only the degree of the vertex needs to be checked.
	\end{proof}
\end{longversion}

Therefore, the solutions to the formula in the \qsatgame{} and the solutions to the online vertex subset game are equivalent.
Thus, the reduction graph $G_\varphi^{\prime\prime\prime}$ is a valid reduction from \qsatgame{} because the one-to-one correspondence between solutions is preserved, which concludes the proof of \Cref{generalReductionCovering:thm:PSPACECompleteness}.
At last, the online algorithm is able to win the game if and only if the \qsatgame{} is winnable.

\begin{restatable}{theorem}{generalReductionCoveringthmPSPACECompleteness}\label{generalReductionCovering:thm:PSPACECompleteness}
	If \VertexCoveringGraphProblem{} is gadget reducible from \sat{} and \Cref{generalReductionCovering:lem:one-to-oneCorrespondence,generalReductionCovering:lem:optimalGamestrategyAdversary,generalReductionCovering:lem:decisionOnVariableGadget} hold, then \OnlineVertexCoveringGraphProblem{} is \PSPACE-complete.
\end{restatable}

\section{Vertex Cover}
\label{vc:sec:vertexCover}
In this section, we use our reduction framework to show that the online vertex subset game based on the \vc{} problem, the \ovg, is \PSPACE-complete.
\vc{} was originally shown to be \NP-complete by Karp \cite{DBLP:conf/coco/Karp72} with a reduction from \textsc{Clique}.
However, since our reduction framework extends reductions from \sat, we use an alternative reduction from Garey and Johnson \cite{garey1979computers}.

Let $\varphi$ be the \sat-formula, let $X$ be the set of $n$ variables and let $C$ be the set of $m$ clauses of $\varphi$.
We construct the following graph $G_\varphi = (V, E)$:
For each variable $x_i$, introduce a variable gadget consisting of two vertices, connected by an edge.
One of these vertices represents the positive literal, while the other represents the negative literal.
Thus we refer to these vertices as literal vertices. 
For each clause $C_j, \indexSet{j}{m}$, we construct a clause gadget, which is a triangle of vertices, where each vertex represents one of the literals in $C_j$.
Finally, each vertex of a clause is connected to the literal it represents.
An example of this construction is shown in \Cref{vc:fig:baseReduction}.

\begin{figure}[!ht]
	\centering
	\scalebox{1}{
	\begin{tikzpicture}[scale=0.5,
						node/.style = {shape=circle, draw, inner sep=0pt, minimum size=0.25cm},
						textnode/.style = {shape=circle, draw, inner sep=0pt, minimum size=0.4cm},
						smallnode/.style = {shape=circle, draw, inner sep=0pt, minimum size=0.1cm},
						box/.style = {rectangle, fill=gray!20, rounded corners, fill opacity=1, inner sep=1pt}]
		\node (bb1) at (-1.3, -3.1) {};
		\node (bb2) at (7.6, 0.6) {};
		\node[box, fill=gray!10, fit=(bb1)(bb2)] (bbox) {};
		\node[above left] (bbtext) at (bbox.south east) {$G_\varphi$};
		
		\node (bt1) at (-1, -0.3) {};
		\node (bf1) at (2.3, 0.3) {};
		\node[box, fit=(bt1)(bf1)] (x1box) {};
		\node[right] (x1text) at (x1box.west) {$x_1$};
		\node[node, thick] (t1) at (0, 0) {};
		\node[node, thick] (f1) at (2, 0) {};
		\path[-, thick] (f1) edge (t1);
		
		\node (bt2) at (4, -0.3) {};
		\node (bf2) at (7.3, 0.3) {};
		\node[box, fit=(bt2)(bf2)] (x2box) {};
		\node[right] (x2text) at (x2box.west) {$x_2$};
		\node[node, thick] (t2) at (5, 0) {};
		\node[node, thick] (f2) at (7, 0) {};
		\path[-, thick] (f2) edge (t2);
		
		\node (bC11) at (2.3, -1.2) {};
		\node (bC12) at (-1, -2.8) {};
		\node[box, fit=(bC11)(bC12)] (C1box) {};
		\node[right] (C1text) at (C1box.west) {$C_1$};
		\node[node, thick] (C11) at (1, -1.5) {};
		\node[node, thick] (C12) at (0, -2.5) {};
		\node[node, thick] (C13) at (2, -2.5) {};
		
		\path[-, thick] (C11) edge (C12);
		\path[-, thick] (C12) edge (C13);
		\path[-, thick] (C13) edge (C11);
		
		\path[-, thick] (C11) edge (f1);
		\path[-, thick] (C12) edge (t1);
		\path[-, thick] (C13) edge (t2);

	\end{tikzpicture}
}
	\caption{The reduction graph for the reduction from \sat{} to \vc{} for instance $\varphi = \left(x_1 \vee \overline{x}_1 \vee x_2\right)$.}
	\label{vc:fig:baseReduction}
\end{figure}

The dependencies in $G_\varphi$ are of the type if a literal vertex is not contained in a solution, then all clause vertices representing the same literal must be contained in that solution.
Therefore, all vertices in $G_\varphi$ are solution dependent.

The \ovg{} has a graph $G$ and a $k \in \mathbb{N}$ as input.
It asks whether there is a winning strategy for the online algorithm, that is, it finds a vertex cover of size at most $k$ for every reveal order while knowing an isomorphic copy of $G$.

\begin{restatable}{theorem}{vcthpspaceComplete}\label{vc:th:pspaceComplete}
	The \ovg{} with the neighborhood reveal model and a map is \PSPACE-complete.
\end{restatable}

The containment of \ovg{} in \PSPACE{} is already established by \Cref{generalReductionCovering:thm:PspaceContainment}.
To show hardness, we extend the above reduction for \vc{} according to our framework.
Therefore, we need to introduce fake clause gadgets, dependency reveal gadgets, and ID gadgets and prove that they fulfill the gadget properties, required by \Cref{generalReductionCovering:lem:one-to-oneCorrespondence,generalReductionCovering:lem:optimalGamestrategyAdversary,generalReductionCovering:lem:decisionOnVariableGadget}. 

\begin{shortversion}
	
	An example for a fake clause gadget is shown in \Cref{vc:fig:baseAndFake}.
	Any optimal vertex cover on the fake clause gadget has size $3$ and contains exactly the triangle representing the clause.
	In the neighborhood reveal model, fake clause gadgets can not be distinguished from real clause gadgets, as long as only vertices of variable gadgets are revealed by the adversary.
	However, as soon as a vertex of the fake clause gadget is revealed, it can be distinguished from a real clause gadget, as the vertex degrees are different.
	
\begin{figure}[!ht]
	\centering
	\scalebox{1}{
	\begin{tikzpicture}[scale=0.5,
						node/.style = {shape=circle, draw, inner sep=0pt, minimum size=0.25cm},
						textnode/.style = {shape=circle, draw, inner sep=0pt, minimum size=0.4cm},
						smallnode/.style = {shape=circle, draw, inner sep=0pt, minimum size=0.1cm},
						box/.style = {rectangle, fill=gray!20, rounded corners, fill opacity=1, inner sep=1pt}]
		\node (bb1) at (-1.3, -3.1) {};
		\node (bb2) at (7.6, 0.6) {};
		\node[box, fill=gray!10, fit=(bb1)(bb2)] (bbox) {};
		\node[above left] (bbtext) at (bbox.south east) {$G_\varphi$};
		
		\node (bt1) at (-1, -0.3) {};
		\node (bf1) at (2.3, 0.3) {};
		\node[box, fit=(bt1)(bf1)] (x1box) {};
		\node[right] (x1text) at (x1box.west) {$x_1$};
		\node[node, thick] (t1) at (0, 0) {};
		\node[node, thick] (f1) at (2, 0) {};
		\path[-, thick] (f1) edge (t1);
		
		\node (bt2) at (4, -0.3) {};
		\node (bf2) at (7.3, 0.3) {};
		\node[box, fit=(bt2)(bf2)] (x2box) {};
		\node[right] (x2text) at (x2box.west) {$x_2$};
		\node[node, thick] (t2) at (5, 0) {};
		\node[node, thick] (f2) at (7, 0) {};
		\path[-, thick] (f2) edge (t2);
		
		\node (bC11) at (2.3, -1.2) {};
		\node (bC12) at (-1, -2.8) {};
		\node[box, fit=(bC11)(bC12)] (C1box) {};
		\node[right] (C1text) at (C1box.west) {$C_1$};
		\node[node, thick] (C11) at (1, -1.5) {};
		\node[node, thick] (C12) at (0, -2.5) {};
		\node[node, thick] (C13) at (2, -2.5) {};
		
		\path[-, thick] (C11) edge (C12);
		\path[-, thick] (C12) edge (C13);
		\path[-, thick] (C13) edge (C11);
		
		\path[-, thick] (C11) edge (f1);
		\path[-, thick] (C12) edge (t1);
		\path[-, thick] (C13) edge (t2);
		
		\node (bC21) at (9.2, -0.7) {};
		\node (bC22) at (13, -2.8) {};
		\node[box, fill=gray!10, fit=(bC21)(bC22)] (C2box) {};
		\node[left] (C2text) at (C2box.east) {\FakeClauseGadget};
		\node[node] (C21) at (11, -1.5) {};
		\node[node] (C22) at (10, -2.5) {};
		\node[node] (C23) at (12, -2.5) {};
		
		\path[-] (C21) edge (C22);
		\path[-] (C22) edge (C23);
		\path[-] (C23) edge (C21);
		
		\path[-, draw=blue, dashed] (C21) edge (f1);
		\path[-, draw=blue, dashed, out=130, in=330, out distance=1cm, in distance=1cm] (C22) edge (t1);
		\path[-, draw=blue, dashed, out=90, in=355, out distance=3cm, in distance=1cm] (C23) edge (f2);
		
		\node[smallnode] (C211) at (10.8, -0.7) {};
		\node[smallnode] (C212) at (11.2, -0.7) {};
		\node[smallnode] (C221) at (9.2, -2.3) {};
		\node[smallnode] (C222) at (9.2, -2.7) {};
		\node[smallnode] (C231) at (12.8, -2.3) {};
		\node[smallnode] (C232) at (12.8, -2.7) {};
		
		\foreach \x in {1,2,3} {
			\foreach \y in {1,2}
			\path[-] (C2\x) edge node[] {} (C2\x\y);
		}
	\end{tikzpicture}
}
	\caption{The reduction graph for the reduction from \sat{} to \vc{} for instance $\varphi = \left(x_1 \vee \overline{x}_1 \vee x_2\right)$.
	The clause $\left(x_1 \vee \overline{x}_1 \vee \overline{x}_2\right)$ does not exist and is represented by a fake clause gadget \FakeClauseGadget.
	The blue dashed edges are the set $E_{con}$ for the fake clause gadget.}
	\label{vc:fig:baseAndFake}
\end{figure}
	The dependency reveal gadget reveals the solution dependencies to the online algorithm. 
	An example is depicted in \Cref{vc:fig:dependency}. 
	Since both the literal vertices and the vertices of clause gadgets are solution dependent, the online algorithm needs to be able to identify which variable they correspond to, and in the case of $\exists$-variables also which literal they correspond to.
	For that, we look at the degrees of all vertices in the graph $G_\varphi^{\prime\prime}$.
	The optimal solution for the dependency reveal gadget always contains exactly the center vertex of the star. 
	
	Finally, we define ID gadgets for literal vertices and all vertices in clause gadgets as they are solution dependent.
	An example is presented in \Cref{vc:fig:Id}. 
	The task of the ID gadget is to enable the algorithm to uniquely  identify vertices of $\exists$-quantified literals and solution dependent vertices as well as to identify the $\forall$-quantified literals such that the literal vertices of one $\forall$-quantified variable are indistinguishable. 
\end{shortversion}

\begin{longversion}
\begin{definition}[Self-contained fake clause gadget for \vc]
	\label{vc:def:fakeClause}
	The fake clause gadget, for non-existing clause $C^\prime_j \notin C$, is a triangle, where each vertex has two additional vertices attached. 
	The vertices of the triangle are connected to the literal vertices they represent.
\end{definition}

An example for a fake clause gadget is shown in \Cref{vc:fig:baseAndFake}.
Any optimal vertex cover on the fake clause gadget has size $3$ and contains exactly the triangle representing the clause.
In the neighborhood reveal model, fake clause gadgets can not be distinguished from real clause gadgets, as long as only vertices of variable gadgets are revealed by the adversary.
However, as soon as a vertex of the fake clause gadget is revealed, it can be distinguished from a real clause gadget, as the vertex degrees are different.

\begin{figure}[!ht]
	\centering
	\scalebox{1}{
	\begin{tikzpicture}[scale=0.5,
						node/.style = {shape=circle, draw, inner sep=0pt, minimum size=0.25cm},
						textnode/.style = {shape=circle, draw, inner sep=0pt, minimum size=0.4cm},
						smallnode/.style = {shape=circle, draw, inner sep=0pt, minimum size=0.1cm},
						box/.style = {rectangle, fill=gray!20, rounded corners, fill opacity=1, inner sep=1pt}]
		\node (bb1) at (-1.3, -3.1) {};
		\node (bb2) at (7.6, 0.6) {};
		\node[box, fill=gray!10, fit=(bb1)(bb2)] (bbox) {};
		\node[above left] (bbtext) at (bbox.south east) {$G_\varphi$};
		
		\node (bt1) at (-1, -0.3) {};
		\node (bf1) at (2.3, 0.3) {};
		\node[box, fit=(bt1)(bf1)] (x1box) {};
		\node[right] (x1text) at (x1box.west) {$x_1$};
		\node[node, thick] (t1) at (0, 0) {};
		\node[node, thick] (f1) at (2, 0) {};
		\path[-, thick] (f1) edge (t1);
		
		\node (bt2) at (4, -0.3) {};
		\node (bf2) at (7.3, 0.3) {};
		\node[box, fit=(bt2)(bf2)] (x2box) {};
		\node[right] (x2text) at (x2box.west) {$x_2$};
		\node[node, thick] (t2) at (5, 0) {};
		\node[node, thick] (f2) at (7, 0) {};
		\path[-, thick] (f2) edge (t2);
		
		\node (bC11) at (2.3, -1.2) {};
		\node (bC12) at (-1, -2.8) {};
		\node[box, fit=(bC11)(bC12)] (C1box) {};
		\node[right] (C1text) at (C1box.west) {$C_1$};
		\node[node, thick] (C11) at (1, -1.5) {};
		\node[node, thick] (C12) at (0, -2.5) {};
		\node[node, thick] (C13) at (2, -2.5) {};
		
		\path[-, thick] (C11) edge (C12);
		\path[-, thick] (C12) edge (C13);
		\path[-, thick] (C13) edge (C11);
		
		\path[-, thick] (C11) edge (f1);
		\path[-, thick] (C12) edge (t1);
		\path[-, thick] (C13) edge (t2);
		
		\node (bC21) at (9.2, -0.7) {};
		\node (bC22) at (13, -2.8) {};
		\node[box, fill=gray!10, fit=(bC21)(bC22)] (C2box) {};
		\node[left] (C2text) at (C2box.east) {\FakeClauseGadget};
		\node[node] (C21) at (11, -1.5) {};
		\node[node] (C22) at (10, -2.5) {};
		\node[node] (C23) at (12, -2.5) {};
		
		\path[-] (C21) edge (C22);
		\path[-] (C22) edge (C23);
		\path[-] (C23) edge (C21);
		
		\path[-, draw=blue, dashed] (C21) edge (f1);
		\path[-, draw=blue, dashed, out=130, in=330, out distance=1cm, in distance=1cm] (C22) edge (t1);
		\path[-, draw=blue, dashed, out=90, in=355, out distance=3cm, in distance=1cm] (C23) edge (f2);
		
		\node[smallnode] (C211) at (10.8, -0.7) {};
		\node[smallnode] (C212) at (11.2, -0.7) {};
		\node[smallnode] (C221) at (9.2, -2.3) {};
		\node[smallnode] (C222) at (9.2, -2.7) {};
		\node[smallnode] (C231) at (12.8, -2.3) {};
		\node[smallnode] (C232) at (12.8, -2.7) {};
		
		\foreach \x in {1,2,3} {
			\foreach \y in {1,2}
			\path[-] (C2\x) edge node[] {} (C2\x\y);
		}
	\end{tikzpicture}
}
	\caption{The reduction graph for the reduction from \sat{} to \vc{} for instance $\varphi = \left(x_1 \vee \overline{x}_1 \vee x_2\right)$.
	The clause $\left(x_1 \vee \overline{x}_1 \vee \overline{x}_2\right)$ does not exist and is represented by a fake clause gadget \FakeClauseGadget.
	The blue dashed edges are the set $E_{con}$ for the fake clause gadget.}
	\label{vc:fig:baseAndFake}
\end{figure}

\begin{restatable}[Gadget Property 1]{lemma}{vclemfakeClauseSelfContained}
	\label{vc:lem:fakeClauseSelfContained}
	The fake clause gadget (\Cref{vc:def:fakeClause}) is self-contained for \vc{}.
\end{restatable}
	\begin{proof}
		Let $C^\prime_j \notin C$ be an arbitrary non-existing clause with $G_{fc}(C^\prime_j)$ being its fake clause gadget.
		Any optimal vertex cover on $G_{fc}(C^\prime_j)$ has size $3$ and contains exactly the triangle. 
		For a contradiction, assume there is a vertex cover $S^\prime$ of size at most $3$ on $G_{fc}(C^\prime_j)$, that does not contain at least one of those three vertices. 
		However, each of them has two neighbors that are not connected to any other vertices, which would then need to be part of the vertex cover instead, while not covering any additional edges. 
		Thus $S^\prime$ additionally needs to contain at least two vertices of the triangle, to cover its edges.
		This is a contradiction to $S^\prime$ containing at most three vertices. 
		
		By the same argument, none of the triangle vertices can be replaced their neighbors in the reduction graph. 
		Additionally, as the triangle is already part of the vertex cover, none of their neighbors in the reduction graph are forced to be in the vertex cover to cover the edge between them. 
		
		Let $G^0_\varphi = G_\varphi$ and let $G^\gamma_\varphi$ be the graph that has been extended with $\gamma$ many fake clause gadgets.
		Then given the graph $G^{\gamma-1}_\varphi$, the graph $G_{fc}(C^\prime_j)$ and the graph $G^\gamma_\varphi = G^{\gamma-1}_\varphi \circ G_{fc}(C^\prime_j)$, the following holds:
		There exists a solution of size $3$ for $G_{fc}(C^\prime_j)$, such that for any optimal solution of size $k^*$ on $G^{\gamma-1}_\varphi$ the disjoint union of those two solutions is an optimal solution of size $k^*+3$ for $G^\gamma_\varphi$.
		Therefore the fake clause gadget from \Cref{vc:def:fakeClause} is self-contained.
	\end{proof}

Since the vertices of a clause gadget are solution dependent on the literal vertices they represent, the dependency reveal gadget needs to account for that. 
An example of a dependency reveal gadget is depicted in \Cref{vc:fig:dependency}. 

\begin{definition}[Self-contained dependency reveal gadget for \vc]\label{vc:def:dependencyReveal}
	The dependency reveal gadget for a $\forall$-variable $x_i$ is a star.
	Its center vertex is connected to the literal vertices of all variables with index at least $i$ (except the true literal of variable $x_i$), and all vertices representing them in clauses (including the true literal of variable $x_i$).
	The number of leaves is such that, together with the connecting edges, the degree of the center vertex equals $3\binom{2n}{3} + 2n + 1$.
\end{definition}
\end{longversion}

\begin{figure}[!ht]
	\begin{subfigure}{0.48\textwidth}
	\centering
	\scalebox{1}{
	\begin{tikzpicture}[scale=0.5,
						node/.style = {shape=circle, draw, inner sep=0pt, minimum size=0.25cm},
						textnode/.style = {shape=circle, draw, inner sep=0pt, minimum size=0.4cm},
						smallnode/.style = {shape=circle, draw, inner sep=0pt, minimum size=0.1cm},
						box/.style = {rectangle, fill=gray!20, rounded corners, fill opacity=1, inner sep=1pt}]
		\node (bb1) at (-1.3, -3.1) {};
		\node (bb2) at (7.6, 0.6) {};
		\node[box, fill=gray!10, fit=(bb1)(bb2)] (bbox) {};
		\node[above left] (bbtext) at (bbox.south east) {$G'_\varphi$};
		
		\node (bt1) at (-1, -0.3) {};
		\node (bf1) at (2.3, 0.3) {};
		\node[box, fit=(bt1)(bf1)] (x1box) {};
		\node[right] (x1text) at (x1box.west) {$x_1$};
		\node[node, thick] (t1) at (0, 0) {};
		\node[node, thick] (f1) at (2, 0) {};
		\path[-, thick] (f1) edge (t1);
		
		\node (bt2) at (4, -0.3) {};
		\node (bf2) at (7.3, 0.3) {};
		\node[box, fit=(bt2)(bf2)] (x2box) {};
		\node[right] (x2text) at (x2box.west) {$x_2$};
		\node[node, thick] (t2) at (5, 0) {};
		\node[node, thick] (f2) at (7, 0) {};
		\path[-, thick] (f2) edge (t2);
		
		\node (bC11) at (2.3, -1.2) {};
		\node (bC12) at (-1, -2.8) {};
		\node[box, fit=(bC11)(bC12)] (C1box) {};
		\node[right] (C1text) at (C1box.west) {$C_1$};
		\node[node, thick] (C11) at (1, -1.5) {};
		\node[node, thick] (C12) at (0, -2.5) {};
		\node[node, thick] (C13) at (2, -2.5) {};
		
		\path[-, thick] (C11) edge (C12);
		\path[-, thick] (C12) edge (C13);
		\path[-, thick] (C13) edge (C11);
		
		\path[-, thick] (C11) edge (f1);
		\path[-, thick] (C12) edge (t1);
		\path[-, thick] (C13) edge (t2);
		\node (bd1) at (8.7, -3.1) {};
		\node (bd2) at (10.3, 0.6) {};
		\node[box, fill=gray!10, fit=(bd1)(bd2)] (bdox) {};
		\node[above right] (bdtext) at (bdox.south west) {$G_{dr}$};
		
		\node[node] (c) at (9, -1.2) {};
		
		\foreach \x in {1,...,10} {
			\node[smallnode] (c\x) at (10, 0.45-0.3*\x) {};
			\path[-] (c) edge node[right] {} (c\x);
		}
		
		\path[-, draw=blue, dashed] (c) edge (f1) edge (t2) edge (f2) edge (C11) edge (C13);
		\path[-, draw=blue, dashed, out=190, in=20, out distance=1cm, in distance=1cm] (c) edge (C12);
	\end{tikzpicture}
	}
\caption{The dependency reveal gadget for the $\forall$-variable $x_1$, with only one variable of higher index, is depicted. 
The fake clause gadgets of $G'_\varphi$ are omitted.}
\label{vc:fig:dependency}
	\end{subfigure}
	\hfill
	\begin{subfigure}{0.48\textwidth}
	\centering
	\scalebox{1}{
	\begin{tikzpicture}[scale=0.5,
		node/.style = {shape=circle, draw, inner sep=0pt, minimum size=0.25cm},
		textnode/.style = {shape=circle, draw, inner sep=0pt, minimum size=0.4cm},
		smallnode/.style = {shape=circle, draw, inner sep=0pt, minimum size=0.1cm},
		box/.style = {rectangle, fill=gray!20, rounded corners, fill opacity=1, inner sep=1pt}]
		\node (g2b1) at (13.8, -3.1) {};
		\node (g2b2) at (17.2, 0.6) {};
		\node[box, fill=gray!10, fit=(g2b1)(g2b2)] (g2box) {};
		\node[above right] (g2text) at (g2box.south west) {$G''_\varphi$};
		
		\node (bt) at (15.7, 0.1) {};
		\node (bf) at (16.9, -2.5) {};
		\node[box, fit=(bt)(bf)] (xbox) {};
		\node[left] (xtext) at (xbox.east) {$x_1$};
		\node[node, thick] (t) at (16, -0.2) {};
		\node[node, thick] (f) at (16, -2.2) {};
		
		\path[-, thick] (f) edge (t);
		
		\foreach \x in {1,...,3} {
			\node (ot\x) at (14.5, 0.6-0.4*\x) {};
			\path[-] (ot\x) edge (t);
		}
		\foreach \x in {1,...,4} {
			\node (of\x) at (14.5, -1.2-0.4*\x) {};
			\path[-] (of\x) edge (f);
		}
		
		\node (bt1) at (18.3, 0.6) {};
		\node (bt2) at (20.2, -1) {};
		\node[box, fill=gray!10, fit=(bt1)(bt2)] (tbox) {};
		\node[left] (ttext) at (tbox.east) {$G_{id}$};
		\foreach \x in {1,...,3} {
			\node[smallnode] (x1t\x) at (18.4, 1.0-0.6*\x) {};
			\node[smallnode] (y1t\x) at (18.9, 1.13-0.6*\x) {};
			\node[smallnode] (z1t\x) at (18.9, 0.87-0.6*\x) {};
			\path[-, draw=blue, dashed] (t) edge (x1t\x);
			\path[-] (x1t\x) edge (y1t\x);
			\path[-] (x1t\x) edge (z1t\x);
		}
		
		\node (bt1) at (18.3, -1.9) {};
		\node (bt2) at (20.2, -2.9) {};
		\node[box, fill=gray!10, fit=(bt1)(bt2)] (tbox) {};
		\node[left] (ttext) at (tbox.east) {$G_{id}$};
		\foreach \x in {1,2} {
			\node[smallnode] (x1f\x) at (18.4, -1.5-0.6*\x) {};
			\node[smallnode] (y1f\x) at (18.9, -1.63-0.6*\x) {};
			\node[smallnode] (z1f\x) at (18.9, -1.37-0.6*\x) {};
			\path[-, draw=blue, dashed] (f) edge (x1f\x);
			\path[-] (x1f\x) edge (y1f\x);
			\path[-] (x1f\x) edge (z1f\x);
		}
	\end{tikzpicture}
	}
\caption{The ID gadgets for the $\forall$-variable $x_1$ are shown. 
	Both literal vertices have the same degree. 
	For each gadget, the blue dashed edges are the set $E_{con}$.}
\label{vc:fig:Id}
	\end{subfigure}
	\caption{Dependency reveal gadget and ID gadget for \vc.}
\end{figure}
\begin{longversion}
The number of leaves of a dependency reveal gadget is always at least $2$.
Given a formula with $n$ variables, there are exactly $2n$ literals, and $\binom{2n}{3}$ possible clauses with three literals.
A dependency reveal gadget can target all three vertices of a clause and all literal vertices, however the true literal of variable $x_i$ is not targeted.
Thus, the number of connecting edges is at most $3\binom{2n}{3} + 2n - 1$.
Therefore, the optimal solution for the dependency reveal gadget always contains exactly the center vertex of the star, as we show in the following lemma. 

\begin{restatable}[Gadget Property 2]{lemma}{vclemdependencyRevealSelfContained}
	\label{vc:lem:dependencyRevealSelfContained}
	The dependency reveal gadget (\Cref{vc:def:dependencyReveal}) is self-contained for \vc{}.
\end{restatable}
	\begin{proof}
		Let $x_i$ be any $\forall$-variable of the \qsatgame-instance.
		Since $G_{dr}(x_i)$ is a star with at least $2$ leaves, the optimal vertex cover on $G_{dr}(x_i)$ contains exactly the center of the star, and no other vertices.
		As covering any edges incident to target vertices by those target vertices has no effect on the edges in $G_{dr}(x_i)$, no solution on the reduction graph can make the optimal solution on $G_{dr}(x_i)$ smaller.
		This also means that the optimal vertex cover on $G_{dr}(x_i)$ already covers all edges connecting it to its target vertices, thus no vertices of the reduction graph are forced to be in the vertex cover by attaching $G_{dr}(x_i)$. 
		
		Remember that $G_\varphi^\prime$ is the graph that resulted from $G_\varphi$ by adding a fake clause gadget for every possible clause $C^\prime_j \notin C$ with three literals.
		Let $G^0_\varphi = G^\prime_\varphi$ and let $G^\gamma_\varphi$ be the graph that has been extended with $\gamma$ many dependency reveal gadgets.
		Then given the graph $G^{\gamma-1}_\varphi$, the graph $G_{dr}(x_i)$ and the graph $G^\gamma_\varphi = G^{\gamma-1}_\varphi \circ G_{dr}(x_i)$, the following holds:
		There exists a solution of size $1$ for $G_{dr}(x_i)$, such that for any optimal solution of size $k^*$ on $G^{\gamma-1}_\varphi$ the disjoint union of those two solutions is an optimal solution of size $k^*+1$ for $G^\gamma_\varphi$.
		Therefore the dependency reveal gadget from \Cref{vc:def:dependencyReveal} is self-contained.
	\end{proof}
\end{longversion}

\begin{longversion}
	Since both the literal vertices and the vertices of clause gadgets are solution dependent, the online algorithm needs to be able to identify which variable they correspond to, and in the case of $\exists$-variables also which literal they correspond to.
	For that, we look at the degrees of all vertices in the graph $G_\varphi^{\prime\prime}$.
	The leaves in the stars of the dependency reveal gadgets and fake clause gadgets have degree $1$, and the variable vertices in fake clause gadgets have degree $5$. 
	The center vertices of the dependency reveal gadgets have degree $3\binom{2n}{3} + 2n$.
	Therefore, any vertex that was not present in $G_\varphi$ already has a degree that is either smaller than $\binom{2n-1}{2} + 4$ or larger than $\binom{2n-1}{2} + 4n+3$ for $n \geq 2$ (which is necessary for three different literals per clause).
	Note that the latter inequality holds due to the following:
	\begin{align*}
	\binom{2n}{3} &> \binom{2n-1}{2} && n \geq 2 \\
	\binom{2n}{3} &\geq 2n && n \geq 2 \\
	3\binom{2n}{3} + 2n &\geq \binom{2n}{3} + 6n > \binom{2n-1}{2} + 4n + 3 && n \geq 2
	\end{align*}
	Therefore, we use that range of degrees for our literal vertices and clause vertices. \\
	Let $d^\forall_<(i)$ (resp. $d^\forall_\leq(i)$) be the number of $\forall$-variables with index smaller (resp. smaller or equal) than $i$.
	Since the formula of the \qsatgame{} always alternates between $\exists$- and $\forall$-variables, $d^\forall_<(i) = \left\lfloor\frac{i}{2}\right\rfloor$ (resp. $d^\forall_\leq(i) = \left\lceil\frac{i}{2}\right\rceil$) if variable $x_1$ is $\forall$-quantified and $d^\forall_<(i) = \left\lfloor\frac{i-1}{2}\right\rfloor$ (resp. $d^\forall_\leq(i) = \left\lceil\frac{i-1}{2}\right\rceil$) otherwise.
	There are $\binom{2n-1}{2}$ possible different clauses with three literals that contain one specific literal. 
	Thus in $G_\varphi^{\prime\prime}$, the literal vertices have degree $\binom{2n-1}{2} + d^\forall_<(i) + 1$, or $\binom{2n-1}{2} + d^\forall_<(i) + 2$ in case of the false literal of a $\forall$-variable that is not the last variable.
	With this, we can define the literal ID gadgets.
	An example for an ID gadget can be seen in \Cref{vc:fig:Id}.
\end{longversion}

\begin{longversion}
	Finally, we define ID gadgets of literal vertices and all vertices in clause gadgets as they are solution dependent.
	Let $d^\forall_<(i)$ (resp. $d^\forall_\leq(i)$) be the number of $\forall$-variables with index smaller (resp. smaller or equal) than $i$.
	Since the formula of the \qsatgame{} always alternates between $\exists$- and $\forall$-quantified variables, $d^\forall_<(i) = \left\lfloor\frac{i}{2}\right\rfloor$ (resp. $d^\forall_\leq(i) = \left\lceil\frac{i}{2}\right\rceil$) if variable $x_1$ is $\forall$-quantified and $d^\forall_<(i) = \left\lfloor\frac{i-1}{2}\right\rfloor$ (resp. $d^\forall_\leq(i) = \left\lceil\frac{i-1}{2}\right\rceil$) otherwise.

\begin{definition}[Self-contained literal ID gadget for \vc]
	\label{vc:def:literalid}
	Let $\ell$ be some literal and $x_i$ its corresponding variable. 
	Let $d^\forall_<(i)$ be defined as above. 
	Further let
	\begin{align*}
		d(\ell) &= 4i - d^\forall_<(i) - 1 && \text{if $\ell$ is positive} \\
		d(\ell) &= 4i - d^\forall_<(i) - 2 && \text{if $\ell$ is negative and $x_i$ is $\forall$-quantified} \\
		d(\ell) &= 4i - d^\forall_<(i) && \text{if $\ell$ is negative and $x_i$ is $\exists$-quantified} 
	\end{align*}
	Then the literal ID gadget for the literal vertex representing $\ell$ consists of $d(\ell)$ paths of length $3$, where each middle vertex is connected to the identified literal vertex.
\end{definition}

It is necessary that the online algorithm can recognize which literal a clause vertex represents, as the adversary could choose to reveal a clause vertex before a literal vertex of the corresponding variable gadget.

\begin{definition}[Self-contained clause ID gadget for \vc]
	\label{vc:def:clauseid}
	Let $\ell$ be a literal with $\ell \in C_j$ for some clause $C_j \in C$ and let $x_i$ be the variable $\ell$ belongs to.
	Let $d^\forall_\leq(i)$ be defined as above.
	Further let $d(\ell) = \binom{2n-1}{2} + 4i - d^\forall_\leq(i) - 1$ if $\ell$ is positive and $d(\ell) = \binom{2n-1}{2} + 4i - d^\forall_\leq(i)$ if $\ell$ is negative.
	Then the clause ID gadget for the clause vertex representing $\ell$ consists of $d(\ell)$ paths of length $3$, where each middle vertex is connected to the identified clause vertex.
\end{definition}

\begin{table}[!ht]
	\begin{tabular}{lll}
		Vertex degree & Type of vertices & Strategy \\
		\hline
		$1$ & leaves of any extension gadget & reject \\
		$3$ & center vertices of ID gadgets & accept \\
		$5$ & triangle vertices of fake clause gadgets & accept \\
		$3\binom{2n}{3} + 2n + 1$ & center vertices of dependency reveal gadgets & accept \\
		$\binom{2n-1}{2} + 4i$ & true literal of $x_i$, also false literal of $x_i$ if it is $\forall$-quantified & depends on $\varphi$ \\
		$\binom{2n-1}{2} + 4i + 1$ & false literal of variable $x_i$ if it is $\exists$-quantified & depends on $\varphi$ \\
		$\binom{2n-1}{2} + 4i + 2$ & vertex representing true literal of variable $x_i$ in a clause & depends on $\varphi$ \\
		$\binom{2n-1}{2} + 4i + 3$ & vertex representing false literal of variable $x_i$ in a clause & depends on $\varphi$ 
	\end{tabular}
	\caption{List of all vertex degrees in the final reduction from \qsatgame{} to \ovn, where $n$ is the number of variables. For any vertex that is not part of a literal or clause gadget, the optimal solution can be deduced just from its degree.}
	\label{vc:tab:degrees}
\end{table}
\newpage
\begin{restatable}[Gadget Property 3]{lemma}{vclemIDGadgetSelfContained}
	\label{vc:lem:IDGadgetSelfContained}
	The ID gadgets (\Cref{vc:def:literalid,vc:def:clauseid}) are self-contained for \vc{}.
\end{restatable}
	\begin{proof}
		Both ID gadgets are a collection of paths with three vertices, where the middle vertex of each path is connected to the identified vertex.
		Thus the size of an optimal vertex cover $S$ on such an ID gadget is equal to the number of paths.
		Assume there is an optimal vertex cover $S^\prime$ on such an ID gadget that does not contain at least one of these middle vertices.
		However, then its two outgoing edges need to be covered, meaning $S^\prime$ is larger than $S$, since the neighbors of the middle vertices can not cover any other edges.
		This is a contradiction to $S^\prime$ being optimal.
		The above argument does not change if the vertex it connects to is in the vertex cover.
		Since the optimal vertex cover on $G_{id}(v)$ already covers all edges connecting it to $v$, $v$ is not forced to be part of the vertex cover by attaching $G_{id}(v)$. 
		
		Let $G^0_\varphi = G^{\prime\prime}_\varphi$ and let $G^\gamma_\varphi$ be the graph that has been extended with $\gamma$ many ID gadgets.
		Then given the graph $G^{\gamma-1}_\varphi$, the graph $G_{id}(v)$ for some literal or clause vertex $v$, and the graph $G^\gamma_\varphi = G^{\gamma-1}_\varphi \circ G_{id}(v)$, the following holds:
		There exists a solution of size $|E_{con}|$ for $G_{id}(v)$, such that for any optimal solution of size $k^*$ for $G^{\gamma-1}_\varphi$ the disjoint union of these two solutions is an optimal solution of size $k^*+|E_{con}|$ for $G^\gamma_\varphi$.
		Therefore, the ID gadgets from \Cref{vc:def:literalid,vc:def:clauseid} are self-contained.
	\end{proof}
\end{longversion}

\begin{longversion}
In our framework, we also add ID gadgets to the fake clause gadgets and dependency reveal gadgets, however the next lemma and \Cref{vc:tab:degrees} show that those vertices can already be recognized by the online algorithm.
Therefore, their ID gadgets are simply the empty graph, which trivially fulfills the property of self-containment.
In \Cref{vc:tab:degrees}, the vertex degrees after adding all ID gadgets are shown.

\begin{restatable}[Gadget Properties 4-6]{lemma}{vclemIDGadgetIdentification}
	\label{vc:lem:IDGadgetIdentification}
	In $G_\varphi^{\prime\prime\prime}$, each solution dependent vertex which is not in a literal gadget of a $\forall$-variable has a unique revelation subgraph. Further, the two literal vertices of a $\forall$-variable have the same revelation subgraph, but different from vertices of any other gadget. Finally, each vertex that is solution independent or part of an extension gadget has a revelation subgraph that allows for an optimal decision.
\end{restatable}
	\begin{proof}
		All solution dependent vertices that are not part of the literal gadget of a $\forall$-variable have a unique revelation subgraph, since they have unique degree, as seen in \Cref{vc:tab:degrees}. 
		
		A literal vertex of a $\forall$-variable shares its degree with exactly one other vertex, that is the vertex representing the negated literal.
		Thus as long as the adversary does not reveal any (fake) clause vertices, dependency reveal vertices or literal vertices of variables with higher index, the two literal vertices of a $\forall$-variable have exactly the same revelation subgraph, but different from any other vertex.
		
		All solution-independent vertices and vertices of extension gadgets that have a degree larger than $1$ are always contained in an optimal solution, by construction of these gadgets. 
		Since their degrees are also always different from solution dependent vertices, the online algorithm can always make an optimal decision for them.
		Finally, by construction of the extension gadgets, no vertex of degree $1$ is contained in any optimal solution. 
		Therefore, the online algorithm can also always make an optimal decision on them based on their revelation subgraph.
	\end{proof}
\end{longversion}

\begin{longversion}
	\paragraph*{Polynomial Time Reduction}
	All our gadgets can be constructed in polynomial time, as they contain at most $\mathcal{O}(\binom{2n}{2}) = \mathcal{O}(n^2)$ many vertices.
	Furthermore, the number of gadgets is also polynomial in the number of variables, as the number of possible clauses with three literals is bounded by $\binom{2n}{3} \in \mathcal{O}(n^3)$.
	Finally, the solution size $k$ can also be computed in polynomial time.
	\begin{enumerate}
		\item For the base reduction, $n + 2m$ vertices are part of the optimal vertex cover.
		\item For the fake clause gadgets, $3 \cdot \left(\binom{2n}{3} - m\right)$ vertices are part of the optimal vertex cover.
		\item Recall that $d^\forall_<(i)$ is the number of $\forall$-variables with an index lower than $i$ and therefore computable in polynomial time. 
		Then the number of vertices in the optimal vertex cover that are part of dependency reveal gadgets is $d^\forall_<(n)$.
		\item For the ID gadgets, we distinguish between literal and clause ID gadgets. For the literal ID gadgets,
		$$
		x^\prime + \sum_{i=1}^{n}2 \cdot \left(4i + d^\forall_<(i) - 1\right)
		$$
		vertices are part of the optimal vertex cover, where $x^\prime = 1$ iff both $x_1$ and $x_n$ are $\exists$-quantified, $x^\prime = -1$ iff both $x_1$ and $x_n$ are $\forall$-quantified, and $x^\prime = 0$ otherwise.
		
		For a variable $x_i$ we denote with $\ell_{2i-1}$ its positive literal and with $\ell_{2i}$ its negative literal.
		For the clause ID gadgets, let $\#(\ell_a)$ be the number of times the literal $\ell_a, ~ a \in \{1,\dots,2n\}$ appears in a clause.
		Then $\sum_{i=1}^{n}(\#(\ell_{2i-1}) + \#(\ell_{2i})) = 3m$ and therefore the size of all $\#(\ell_a)$ is polynomially bounded. 
		The number of vertices in the optimal vertex cover, that are part of a clause ID gadget, is given by
		$$
		\sum_{i=1}^{n} \#(\ell_{2i-1}) \cdot \left(\binom{2n-1}{2} + 4i - d^\forall_\leq(i) - 1\right) + \#(\ell_{2i}) \cdot \left(\binom{2n-1}{2} + 4i - d^\forall_\leq(i)\right)
		$$
	\end{enumerate}
\end{longversion}

Since all constructions are polynomial time computable, we established the requirements for \Cref{generalReductionCovering:thm:PSPACECompleteness}, thus \Cref{vc:th:pspaceComplete} is proven.
The full construction of $G_\varphi^{\prime\prime\prime}$ is shown in \Cref{vc:fig:complete}.

\input{tikz/VC_complete}

\section{More Vertex Subset Problems}
In this section, we apply \Cref{generalReductionCovering:thm:PSPACECompleteness} to the \oig{} and \odg{}.
Like the \ovg{}, they take a graph $G$ and a number $k \in \mathbb{N}$ as input. 
They ask whether there is a winning strategy for the online algorithm, that is, it finds an independent set (resp. dominating set) of size at least (resp. most) $k$ for every reveal order while knowing an isomorphic copy of $G$.

\begin{restatable}{theorem}{isthpspaceComplete}\label{is:th:pspaceComplete}
	The \oig{} with the neighborhood reveal model and a map is \PSPACE-complete.
\end{restatable}
\begin{longversion}
	\begin{proof}
		We again extend an existing reduction and apply \Cref{generalReductionCovering:thm:PSPACECompleteness}.
		
		For the base reduction from \sat{} to \is, we use a slight modification of the reduction from \sat{} to \vc{} given in \cite{garey1979computers}.
		Instead of connecting each clause vertex to the literal vertex it represents, we connect it to its negation.
		The size of the independent set that should be found in $G_\varphi$ is then $k = |X| + |C|$.
		The correctness argument works analogously to the reduction for \vc.
		
		For the fake clause gadget, we use the same construction as for \vc, but the target vertices are adjusted in the same way as for the clause gadgets.
		In case of the dependency reveal gadget and ID gadgets we use exactly the same constructions as for \vc.
		The full construction of the reduction is shown in \Cref{is:fig:complete}.
		Recall that for any vertex cover $S \subseteq V$, the set $V \setminus S$ forms an independent set.
		Since the optimal solution of any of our extension gadgets for \vc{} are the vertices incident to $E_{con}$ of that gadget, the optimal solution of these gadgets for \is{} contains all vertices not incident to $E_{con}$.
		Thus, their optimal solution can not be influenced by the solution on $G_\varphi$ and vice versa.
		Therefore, they are self-contained for \is.
		Gadget Properties $4$-$6$ hold for \is{} by the same arguments as for \vc.
		
		Since this modified reduction is obviously still computable in polynomial time, \Cref{generalReductionCovering:thm:PspaceContainment,generalReductionCovering:thm:PSPACECompleteness} prove our claim.
	\end{proof}
	
	\input{tikz/IS_complete}
\end{longversion}

\begin{restatable}{theorem}{dsthpspaceComplete}\label{ds:th:pspaceComplete}
	The \odg{} with the neighborhood reveal model and a map is \PSPACE-complete.
\end{restatable}
\begin{longversion}
	We again extend an existing reduction and apply \Cref{generalReductionCovering:thm:PSPACECompleteness}.
	For the base reduction, we use the following folklore reduction.
	
	Let $\varphi$ be the \sat-formula, let $X$ be set of $n$ variables and let $C$ be the set of $m$ clauses of $\varphi$. 
	We construct the following graph $G_\varphi = (V, E)$:
	For each variable $x_i, \indexSet{i}{n}$, we introduce a triangle as its variable.
	Two of the vertices represent the literals, which we call literal vertices.
	The third vertex ensures that always one vertex of the variable gadget has to be chosen. 
	For each clause $C_j \indexSet{j}{m}$, we construct a clause gadget consisting of the single vertex.
	Each clause vertex is connected to the literal vertices of the literals the clause contains.
	The size of the dominating set that should be found in $G_\varphi$ is set to $k = n$.
	
	The logical dependencies in $G_\varphi$ are of the type if the literal vertex representing the literal $\ell$ is not contained in a solution, then the literal vertex representing $\overline{\ell}$ must be contained.
	Furthermore, for each clause containing $\ell$, one of the remaining literal vertices must be contained.
	As any truth assignment assigns each literal either true or false, the third vertex of the variable gadget is solution independent.
	Therefore, only the literal vertices in $G_\varphi$ are solution dependent.
	
	We define the following extension gadgets for \ds.
	\begin{definition}[Self-contained fake clause gadget for \ds]
		\label{ds:def:fakeclause}
		The fake clause gadget for non-existing clause $C^\prime_j \notin C$, is a star with $2n-2$ leaves.
		Its center is connected to the literal vertices representing the literals contained in the clause.
	\end{definition}
	\begin{definition}[Self-contained dependency reveal gadget for \ds]
		\label{ds:def:dependencyReveal}
		The dependency reveal gadget for a $\forall$-variable $x_i$ is a star.
		The target vertices are the literal vertices of all variables with index at least $i$, except the true literal of variable $x_i$.
		The number of leaves is such that, together with the connecting edges, the degree of the center vertex equals $2n+1$.
	\end{definition}
	\begin{definition}[Self-contained literal ID gadget for \ds]
		\label{ds:def:literalid}
		Let $\ell$ be some literal, and $x_i$ the corresponding variable. 
		Let $d^\forall_<(i)$ be defined as above. Further let
		\begin{align*}
			d(\ell) &= 4i - d^\forall_<(i) - 2 && \text{if $\ell$ is positive} \\
			d(\ell) &= 4i - d^\forall_<(i) - 3 && \text{if $\ell$ is negative and $x_i$ is $\forall$-quantified} \\
			d(\ell) &= 4i - d^\forall_<(i) - 1 && \text{if $\ell$ is negative and $x_i$ is $\exists$-quantified} 
		\end{align*}
		Then the literal ID gadget for the literal vertex representing $\ell$ is a star with $\binom{2n-1}{2} + 4(n+1)$ leaves, where $d(\ell)$ of those leaves are also connected to the literal vertex.
	\end{definition}
	
	Our fake clause gadget and dependency reveal gadget for \ds{} have the same structure, as they are stars. 
	Any optimal dominating set on those gadgets contains exactly the center vertex of the star.
	The connection of those gadgets to the reduction graph is done only by edges from the centers of the respective stars to literal vertices.
	Thus, neither can make the solution on the reduction graph smaller, as the literal vertices of each variable are in a triangle together with a third vertex that is not connected to any other vertex.
	At the same time, dominating the center vertex of one of those gadgets by one of the connected literal vertices never removes it from an optimal dominating set, as it has at least two leaves.
	This proves the Gadget Properties $1$ and $2$.
	
	The literal ID gadget is also a star.
	Since there are always at least $2$ of its leaves that are not connected to any other vertex, an optimal solution on the ID gadget always contains the center vertex.
	Thus the vertices that are connected to the reduction graph are dominated, but not contained in any optimal dominating set.
	Since they are all connected to the same vertex, this proves Gadget Property $3$.
	
	After extending $G_\varphi$ with the fake clause gadgets, dependency reveal gadgets and literal ID gadgets (in that order), Gadget Properties $4$-$6$ already hold as shown in \Cref{ds:tab:degrees}.
	Thus we can use the empty graph as an ID gadget for the clauses, fake clause gadgets and dependency reveal gadgets to obtain a reduction according to our framework.
	Since the reduction can be computed in polynomial time by an analogous argument to the reduction for \ovg, \Cref{generalReductionCovering:thm:PSPACECompleteness} proves \PSPACE-hardness of \odg.
	Thus together with \Cref{generalReductionCovering:thm:PspaceContainment}, it is also \PSPACE-complete.
	
	\begin{table}
		\begin{tabular}{lll}
			Vertex degree & Type of vertices & Strategy \\
			\hline
			$1$ & leaves of any fake clause / dependency reveal gadget, & reject \\
			& leaves of literal ID gadgets not adjacent to a literal vertex & \\
			$2$ & leaves of literal ID gadgets adjacent to a literal vertex & reject \\
			& third vertex of variable gadgets &  \\
			$3$ & clause vertices & reject \\
			$2n+1$ & center vertex of fake clause / dependency reveal gadgets & accept \\
			$\binom{2n-1}{2} + 4(n+1)$ & center vertex of any literal ID gadget & accept \\
			$\binom{2n-1}{2} + 4i$ & true literal of $x_i$, also false literal of $x_i$ if it is $\forall$-quantified & depends on $\varphi$ \\
			$\binom{2n-1}{2} + 4i + 1$ & false literal of variable $x_i$ if it is $\exists$-quantified & depends on $\varphi$ 
		\end{tabular}
		\caption{List of all vertex degrees in the final reduction from \qsatgame{} to \odg. For any vertex that is not part of a literal gadget, the optimal solution can be deduced just from its degree.}
		\label{ds:tab:degrees}
	\end{table}
	
\begin{figure}[!ht]
	\centering
	\begin{tikzpicture}[scale=0.5,
						node/.style = {shape=circle, draw, inner sep=0pt, minimum size=0.15cm},
						largenode/.style = {shape=circle, draw, inner sep=0pt, minimum size=0.4cm},
						smallnode/.style = {shape=circle, draw, inner sep=0pt, minimum size=0.1cm},
						box/.style = {rectangle, fill=gray!20, rounded corners, fill opacity=1, inner sep=1pt}]
		\def\a{0}
		\def\b{0}
		\node (bx11) at (\a-0.5, \b-0.5) {};
		\node (bx12) at (\a+3.5, \b+2) {};
		\node[box, fit=(bx11)(bx12)] (boxx1) {};
		\node[below right] (x1text) at (boxx1.north west) {$\forall x_1$};
		\node[largenode, ultra thick, fill=\colorno] (x1h) at (\a+1.5, \b+1.5) {};
		\node[largenode, ultra thick, fill=\colormaybe] (x1t) at (\a, \b) {\tiny$t$};
		\node[largenode, ultra thick, fill=\colormaybe] (x1f) at (\a+3, \b) {\tiny$f$};
		\path[-, ultra thick] (x1h) edge (x1f) edge (x1t) (x1f) edge (x1t);
		
		\node (bx1f1) at (\a+4.5, \b-2) {};
		\node (bx1f2) at (\a+6, \b+3.2) {};
		\node[box, fit=(bx1f1)(bx1f2)] (boxx1f) {};
		\node[below right] (x1ftext) at (boxx1f.north west) {\IdGadget};
		\node[node, fill=\coloryes] (x1fc) at (\a+5, \b) {};
		\node[smallnode, fill=\colorno] (x1f1) at (\a+4.5, \b) {};
		\path[-] (x1f1) edge (x1fc);
		\path[-, draw=\colorconnect, dashed] (x1f1) edge (x1f);
		\foreach \x in {1,...,14} {
			\node[smallnode, fill=\colorno] (x1of\x) at (\a+6, \b+3.6-0.4*\x) {};
			\path[-] (x1of\x) edge (x1fc);
		}
		
		\node (bx1t1) at (\a-3, \b-2) {};
		\node (bx1t2) at (\a-1.5, \b+3.2) {};
		\node[box, fit=(bx1t1)(bx1t2)] (boxx1t) {};
		\node[below left] (x1ttext) at (boxx1t.north east) {\IdGadget};
		\node[node, fill=\coloryes] (x1tc) at (\a-2, \b) {};
		\foreach \x in {1,2} {
			\node[smallnode, fill=\colorno] (x1t\x) at (\a-1.5, \b+0.6-0.4*\x) {};
			\path[-] (x1t\x) edge (x1tc);
			\path[-, draw=\colorconnect, dashed] (x1t\x) edge (x1t);
		}
		\foreach \x in {1,...,13} {
			\node[smallnode, fill=\colorno] (x1ot\x) at (\a-3, \b+3.6-0.4*\x) {};
			\path[-] (x1ot\x) edge (x1tc);
		}
		
		\def\a{12}
		\def\b{0}
		\node (bx21) at (\a-0.5, \b-0.5) {};
		\node (bx22) at (\a+3.5, \b+2) {};
		\node[box, fit=(bx21)(bx22)] (boxx2) {};
		\node[below left] (x2text) at (boxx2.north east) {$\exists x_2$};
		\node[largenode, ultra thick, fill=\colorno] (x2h) at (\a+1.5, \b+1.5) {};
		\node[largenode, ultra thick, fill=\colormaybe] (x2t) at (\a, \b) {\tiny$t$};
		\node[largenode, ultra thick, fill=\colormaybe] (x2f) at (\a+3, \b) {\tiny$f$};
		\path[-, ultra thick] (x2h) edge (x2f) edge (x2t) (x2f) edge (x2t);
		
		\node (bx2f1) at (\a+4.5, \b-2) {};
		\node (bx2f2) at (\a+6, \b+2.2) {};
		\node[box, fit=(bx2f1)(bx2f2)] (boxx2f) {};
		\node[below right] (x2ftext) at (boxx2f.north west) {\IdGadget};
		\node[node, fill=\coloryes] (x2fc) at (\a+5, \b) {};
		\foreach \x in {1,...,6} {
			\node[smallnode, fill=\colorno] (x2f\x) at (\a+4.5, \b+1.2-0.4*\x) {};
			\path[-] (x2f\x) edge (x2fc);
			\path[-, draw=\colorconnect, dashed] (x2f\x) edge (x2f);
		}
		\foreach \x in {1,...,9} {
			\node[smallnode, fill=\colorno] (x2of\x) at (\a+6, \b+2-0.4*\x) {};
			\path[-] (x2of\x) edge (x2fc);
		}
		
		\node (bx2t1) at (\a-3, \b-2) {};
		\node (bx2t2) at (\a-1.5, \b+2.2) {};
		\node[box, fit=(bx2t1)(bx2t2)] (boxx2t) {};
		\node[below left] (x2ttext) at (boxx2t.north east) {\IdGadget};
		\node[node, fill=\coloryes] (x2tc) at (\a-2, \b) {};
		\foreach \x in {1,...,5} {
			\node[smallnode, fill=\colorno] (x2t\x) at (\a-1.5, \b+1.2-0.4*\x) {};
			\path[-] (x2t\x) edge (x2tc);
			\path[-, draw=\colorconnect, dashed] (x2t\x) edge (x2t);
		}
		\foreach \x in {1,...,10} {
			\node[smallnode, fill=\colorno] (x2ot\x) at (\a-3, \b+2-0.4*\x) {};
			\path[-] (x2ot\x) edge (x2tc);
		}
		
		\def\a{3}
		\def\b{4.5}
		\node (bd1) at (\a-1.8, \b-0.5) {};
		\node (bd2) at (\a+0.5, \b+1) {};
		\node[box, fit=(bd1)(bd2)] (boxd) {};
		\node[right] (dtext) at (boxd.west) {\DependencyRevealGadget};
		\node[largenode, fill=\coloryes] (drg1) at (\a, \b) {};
		\node[smallnode, fill=\colorno] (drg11) at (\a-0.3, \b+1) {};
		\node[smallnode, fill=\colorno] (drg12) at (\a+0.3, \b+1) {};
		\path[-] (drg1) edge (drg11) edge (drg12);
		
		\path[-, draw=\colorconnect, dashed] (x1f) edge (drg1);
		\path[-, draw=\colorconnect, dashed, in=0, out=110, in distance=5cm, out distance=4.5cm] (x2f) edge (drg1);
		\path[-, draw=\colorconnect, dashed, in=355, out=105, in distance=4cm, out distance=4cm] (x2t) edge (drg1);
		
		\def\a{1.5}
		\def\b{-5.5}
		\node (bc11) at (\a-1.5, \b-0.5) {};
		\node (bc12) at (\a+0.5, \b+0.5) {};
		\node[box, fit=(bc11)(bc12)] (boxc1) {};
		\node[right] (c1text) at (boxc1.west) {$C_1$};
		\node[largenode, ultra thick, fill=\colorno] (C1) at (\a, \b) {};
		\path[-, ultra thick] (x1t) edge (C1);
		\path[-, ultra thick] (x1f) edge (C1);
		\path[-, ultra thick] (x2f) edge (C1);
		
		\def\a{13.5}
		\def\b{-5.5}
		\node (bc31) at (\a-1.5, \b-0.5) {};
		\node (bc32) at (\a+0.5, \b+0.5) {};
		\node[box, fit=(bc31)(bc32)] (boxc3) {};
		\node[right] (c3text) at (boxc3.west) {$C_2$};
		\node[largenode, ultra thick, fill=\colorno] (C3) at (\a, \b) {};
		\path[-, ultra thick] (x1t) edge (C3);
		\path[-, ultra thick] (x2t) edge (C3);
		\path[-, ultra thick] (x2f) edge (C3);
		
		\def\a{5.5}
		\def\b{-5.5}
		\node (bc21) at (\a-1.7, \b-1) {};
		\node (bc22) at (\a+0.5, \b+0.5) {};
		\node[box, fit=(bc21)(bc22)] (boxc2) {};
		\node[right] (c2text) at (boxc2.west) {\FakeClauseGadget};
		\node[largenode, fill=\coloryes] (C2) at (\a, \b) {};
		\path[-, draw=\colorconnect, dashed] (x1f) edge (C2);
		\path[-, draw=\colorconnect, dashed] (x2t) edge (C2);
		\path[-, draw=\colorconnect, dashed] (x2f) edge (C2);
		
		\node[smallnode, fill=\colorno] (C21) at (\a-0.3, \b-1) {};
		\node[smallnode, fill=\colorno] (C22) at (\a+0.3, \b-1) {};
		\path[-] (C2) edge(C21) edge (C22);
		
		\def\a{9.5}
		\def\b{-5.5}
		\node (bc41) at (\a-1.7, \b-1) {};
		\node (bc42) at (\a+0.5, \b+0.5) {};
		\node[box, fit=(bc41)(bc42)] (boxc4) {};
		\node[right] (c4text) at (boxc4.west) {\FakeClauseGadget};
		\node[largenode, fill=\coloryes] (C4) at (\a, \b) {};
		\path[-, draw=\colorconnect, dashed] (x1t) edge (C4);
		\path[-, draw=\colorconnect, dashed] (x1f) edge (C4);
		\path[-, draw=\colorconnect, dashed] (x2t) edge (C4);
		
		\node[smallnode, fill=\colorno] (C41) at (\a-0.3, \b-1) {};
		\node[smallnode, fill=\colorno] (C42) at (\a+0.3, \b-1) {};
		\path[-] (C4) edge (C41) edge (C42);
	\end{tikzpicture}
	\caption{Complete view on the reduction for the \qsat{}-instance $\forall x_1 \exists x_2 \left(x_1 \vee \overline{x_1} \vee \overline{x_2}\right) \wedge \left(x_1 \vee x_2 \vee \overline{x_2}\right)$ to \odg{}. The thick vertices and edges represent the original reduction. The \colorconnect{} dashed edges are the connecting edges of the extension gadgets. There are optimal solutions that contain all the \coloryes{} vertices and none of the \colorno{} vertices. Whether the \colormaybe{} vertices are contained depends on the feasible solutions for the \qsat-formula.}
	\label{ds:fig:complete}
\end{figure}
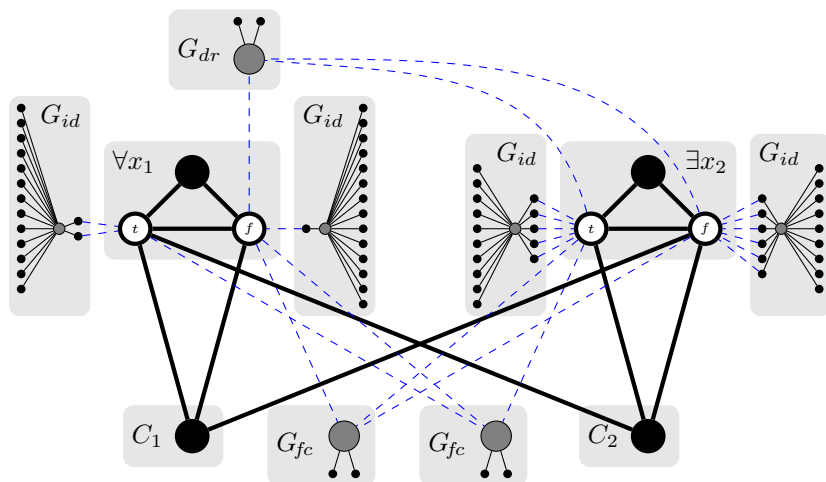
\end{longversion}

\section{Conclusion}
We derived online games from the typical online setting in order to analyze their computational complexity.
Furthermore, we developed a framework for online versions of vertex subset problems with neighborhood reveal model that allows reductions from \qsatgame~ to show that these are \PSPACE-complete.
We showed particularly that the online versions \vc, \is~ and \ds~with neighborhood reveal model are \PSPACE-complete.

The gap between the complexity analysis that started with the \NP-hardness for the problems from AOC under the vertex arrival model and our \PSPACE-completeness results need to be closed. 
From our results the questions arises if the three problems \vc{}, \is{} and \ds{} are actually \PSPACE-complete under the vertex arrival model. 
One way to show the \PSPACE-completeness is to use our reduction framework and add a type of error correction gadget. 
However, the missing knowledge in the vertex arrival model might increase the asymmetry in favor of the adversary such that the complexity decreases and it remains \NP-hard.  

Additionally, the presented framework may be extended to more general subset problems from AOC where the solution is not a vertex subset.  



\newpage

\bibliography{bibliography}

\newpage

\begin{shortversion}
	\appendix
	\input{src/appendix.tex}
\end{shortversion}

\end{document}